\documentclass[runningheads]{llncs}
\usepackage[T1]{fontenc}
\usepackage{cite}
\usepackage{amsmath,amssymb,amsfonts}
\usepackage{algorithmic}
\usepackage{graphicx}
\usepackage{textcomp}
\usepackage{xcolor}
\usepackage{multirow}
\usepackage{mathpartir}

\newcommand{\blind}[1]{}

\renewcommand{\phi}{\varphi}
\renewcommand{\epsilon}{\varepsilon}

\newcommand{\maji}[1]{\ensuremath{\mathbb{#1}}}
\newcommand{\bigfun}[1]{\ensuremath{\mathcal{#1}}}

\newcommand{\R}{{\maji{R}}}
\newcommand{\Rnonneg}{{\maji{R}}_{\geq 0}}
\newcommand{\Rpos}{{\maji{R}}_{> 0}}
\newcommand{\Z}{{\maji{Z}}} 
 
\newcommand{\Znonneg}{\maji{Z}_{\geq 0}}
\newcommand{\N}{\Znonneg}

\newcommand{\OOO}{\maji{O}}
\newcommand{\AAA}{\mathcal{A}}
\newcommand{\BBB}{\mathcal{B}}

\newcommand{\FFF}{\mathcal{F}}
\newcommand{\GGG}{\bigfun{G}}
\newcommand{\LLL}{\mathcal{L}}
\newcommand{\SSS}{{\bigfun{S}}}
\newcommand{\PPP}{\mathcal{P}}

\newcommand{\xto}[1]{\xrightarrow{#1}}

\newcommand{\Traj}{\mathit{Traj}}

\newcommand{\ini}{\textnormal{in}}

\DeclareMathOperator{\AP}{AP}

\newcommand{\Coloneqq}{\mathrel{\mathop{::}=}}
\newcommand{\gramor}{\,|\,}

\DeclareMathOperator{\Next}{\bigcirc}
\DeclareMathOperator{\Until}{U}
\DeclareMathOperator{\Release}{R}
\DeclareMathOperator{\Finally}{\Diamond}
\DeclareMathOperator{\Globally}{\Box}

\newcommand{\chopping}[2]{\ensuremath{[{#1}]_{#2}}}
\newcommand{\actualchopping}[2]{\ensuremath{[{#1}]_{#2}}} 
\DeclareMathOperator{\sub}{sub}

\usepackage{amsmath}

\usepackage{amsthm}
\usepackage{amssymb} 
\usepackage{amsfonts}
\usepackage{stmaryrd}

\newtheorem{asm}{Assumption}
\newcommand{\switchversion}[2]{#2}

\newcommand{\set}[1]{\left\{{#1}\right\}}
\newcommand{\setcomp}[2]{\left\{{#1}\,\middle|\,{#2}\right\}}

\DeclareMathOperator{\cs}{cs}

\usepackage{todonotes}
\usepackage{xcolor, soul}
\usetikzlibrary{patterns,shapes,snakes}
\usepackage{tikz}
\tikzstyle{block} = [draw, fill=white, rectangle, 
    minimum height=3em, minimum width=6em]
\tikzstyle{sum} = [draw, fill=white, circle, node distance=1cm]
\tikzstyle{input} = [coordinate]
\tikzstyle{output} = [coordinate]
\tikzstyle{pinstyle} = [pin edge={to-,thin,black}]
\usetikzlibrary{decorations.pathmorphing}
\usetikzlibrary{automata, positioning}
\tikzset{snake it/.style={decorate, decoration=snake}}
\usetikzlibrary{backgrounds}

\usepackage[linesnumbered,ruled,vlined]{algorithm2e}
\usepackage{hyperref}
\usepackage{cleveref}
\crefname{algorithm}{Alg.}{Algs.}
\newcommand{\abs}{\gamma}
\newcommand{\ABZESN}{\text{$\mathit{AP}$-observation}}

\usepackage{mathabx}
\usepackage{array}
\usepackage{orcidlink}

\newcommand{\infinitynorm}[1]{\|{#1}\|_\infty }
\newcommand{\absolutevalue}[1]{\lVert{#1}\rVert}

\newcommand{\apgreen}{\ensuremath{\mathit{g}}}
\newcommand{\apred}{\ensuremath{\mathit{r}}}
\newcommand{\apcyan}{\ensuremath{\mathit{c}}}
\newcommand{\apblue}{\ensuremath{\mathit{b}}}
\newcommand{\appink}{\ensuremath{\mathit{p}}}

\newcommand{\strat}{\pi}
\newcommand{\GameF}{F_{g}}

\begin{document}
\title{\switchversion{\ABZESN{} Automata for Abstraction-based Verification of Continuous-time Systems}{\ABZESN{} Automata for Abstraction-based Verification of Continuous-time Systems (Extended Version)}%
%\blind{
\thanks{\switchversion{S. Pruekprasert is supported by JSPS KAKENHI Grant
Numbers JP21K14191 and JP22KK0155 and C. Eberhart is supported
by JSPS KAKENHI Grant Number JP25H00446.
This work was partly done while S. Pruekprasert was %affiliated  
with the National Institute of Advanced Industrial Science and Technology, Tokyo, Japan, and C. Eberhart was %affiliated
with the National Institute of Informatics, Tokyo, Japan.
\\
Both authors have contributed equally.}{S. P. is supported by JSPS KAKENHI Grant
Numbers JP21K14191 and JP22KK0155 and C. E. %is supported
by JSPS KAKENHI Grant Number JP25H00446.
This work was partly done while S. P. was %affiliated  
with the National Institute of Advanced Industrial Science and Technology, Tokyo, Japan, and C. E. %was %affiliated
with the National Institute of Informatics, Tokyo, Japan.
% \\
Both authors have contributed equally.}
}}

\titlerunning{\ABZESN{} Automata for Abstraction-based Verification}

\author{
% Anonymous Authors
Sasinee Pruekprasert\inst{1}\orcidlink{0000-0002-5929-9014}
\and
Clovis Eberhart\inst{2}\orcidlink{0000-0003-3009-6747}
\institute{The University of Tokyo, Tokyo, Japan
\email{spruekprasert@g.ecc.u-tokyo.ac.jp}
\and
Tohoku University, Sendai, Japan
\email{eberhart.clovis.d1@tohoku.ac.jp}}
}
% \authorrunning{Anonymous Authors}
\authorrunning{S. Pruekprasert and C. Eberhart}

\maketitle              % typeset the header of the contribution
\begin{abstract}
A key challenge in abstraction-based verification and control under complex specifications such as Linear Temporal Logic (LTL) is that abstract models retain significantly less information than their original systems.
This issue is especially true for continuous-time systems, where the system state trajectories are split into intervals of discrete actions, and satisfaction of atomic propositions is abstracted to a whole time interval.
%\ssc{
% To tackle this challenge, this work introduces a novel translation of LTL specifications to a particular type of Büchi automata, which we call \ABZESN{} automata, specifically designed for abstraction-based verification under LTL specifications.
To tackle this challenge, this work introduces a novel translation from LTL specifications to \ABZESN{} automata, 
a particular type of Büchi automata specifically designed for abstraction-based verification.
% To tackle this challenge, this work introduces \ABZESN{} automata, a type of Büchi automata 
% specifically designed for abstraction-based verification %and symbolic control 
% under LTL specifications.
% ABZESN Automata capture the abstract properties of atomic propositions along system trajectories, providing an innovative approach to handling complex specifications.
% 6-LTL, a logic that abstracts the satisfaction of LTL formulas on a time interval to a set of six possible patterns.
% In order to verify LTL formulas on concrete systems, we give a reduction of 6-LTL to ABZESN automata, a novel construction that mirrors the translation from LTL to Büchi automata.
Based on this automaton, we present a game-based verification algorithm played
between the system and the environment,
% between two players, representing angelic and demonic types of nondeterminism,
and an illustrative example for abstraction-based system verification under several LTL specifications. 
% We also present a verification algorithm and an illustrative example for abstraction-based system verification under LTL specification utilizing the proposed structure.
%}
\end{abstract}

% \freefootnote{${}^*$ Equal contribution.}

% \todo[inline]{change the abstract}

\section{Introduction}

The growing complexity of engineered physical systems has increased the need for formal methods that can specify and verify their desired behaviors.
Many such properties can only be specified in temporal logics, %which has proved to be 
a powerful framework for formalizing complex specifications of timed systems.
In particular, \emph{Linear Temporal Logic} (LTL)~\cite{pnueli1977temporal} strikes a good balance between expressivity and complexity of verification.
% Indeed, many temporal logics such as CTL and CTL$^*$ can be translated to LTL~\cite{}\todo{I thought CTL* was strictly more expressive? (not sure)},
Indeed, LTL is widely used  
for describing temporal specifications in many fields, such as 
%program verification~\cite{pnueli2005applications}, 
verification~\cite{
kern1999formal, pnueli2005applications, rozier2011linear
}
% concurrency~\cite{lynch1996distributed, chai2014online}, 
and control theory~\cite{belta2007symbolic, belta2019formal,  banerjee2025challenges}, thanks to its expressivity.
% At the same time‚ 
Verification of LTL properties can be reduced to language emptiness of Büchi automata~\cite{vardi1986automata,gerth1995simple,gastin2001fast}, which gives rise to efficient verification algorithms. 
However, these techniques are developed for discrete-time, discrete-state systems, while physical systems evolve in continuous space and time.

% On the other hand, In control theory, numerous techniques have been developed to analyze and design controllers for continuous-time, continuous-state systems~\cite{ogata2009modern, khalil2002nonlinear}. 
On the other hand, traditional control theory provides a wealth of methods for analyzing and designing controllers for continuous-time, continuous-state systems~\cite{ogata2009modern, khalil2002nonlinear}.
However, they focus primarily on specifications such as stability~\cite{gu2003survey}, robustness~\cite{sastry2011adaptive}, 
and safety constraints~\cite{ames2016control}. 
% However, they focus primarily on specifications such as stability~\cite{gu2003survey}
% , robustness~\cite{sastry2011adaptive}, disturbance rejection~\cite{guo2014anti}, and tracking~\cite{li2000survey}.
% Lyapunov-based methods are widely used to ensure stability~\cite{lyapunov1992general}, while control barrier functions have been introduced to enforce safety constraints~\cite{ames2016control}.
% Using these two techniques, it is possible to verify so-called \emph{reach-avoid} specifications.
% However, 
Meanwhile, modern applications, such as autonomous systems, require temporal and logic-based properties~\cite{doherty2009temporal, jha2018safe, arechiga2019specifying}, which conventional control methods are not designed for.
% Although other temporal logics exist for continuous-time systems (e.g., MTL~\cite{MTL}, MITL~\cite{MITL}, 
% STL~\cite{STL}), LTL is a natural candidate as a specification language for continuous systems, using a \emph{signal semantics} that operates in continuous time.

\emph{Abstraction-based control}, or \emph{symbolic control}~\cite{
%egerstedt2006special,
tabuada2007approximate, 
pola2008approximately, 
tabuada2008approximate
}, offers a framework to handle complex specifications by constructing a discrete abstraction, called a \emph{symbolic model}, of the original continuous system.  
This approach allows us to leverage automata-theoretic techniques to prove properties of continuous systems.
% Although various abstraction-based frameworks have been explored,
Recent studies on abstraction-based frameworks~\cite{%hsu2018lazy,macoveiciuc2019memory,  
meyer2019hierarchical,
%bai2019incremental,
pruekprasert2021fast, macoveiciuc2022fly, ren2024zonotope} have primarily addressed computational complexity and adaptability, key challenges in this domain.
This work focuses on another fundamental issue in abstraction-based approaches:  the substantial loss of information in abstract models relative to their concrete counterparts. 
This is especially true for continuous-time systems, where trajectories are partitioned into discrete intervals and atomic proposition satisfaction is abstracted over these intervals, complicating verification under complex specifications like LTL.
%Similarly, timed and hybrid automata~\cite{alur1994theory, alur1991hybrid} abstract continuous-time systems into discrete-state machines amenable to verification.

% Other frameworks are actively pursued for verification of continuous-time systems, in particular timed 
% automata~\cite{alur1994theory} and hybrid automata~\cite{alur1991hybrid} for cyber-physical systems (CPSs).
% Recent works on timed automata include parametric extensions of timed automata~\cite{andre2019s} or security of timed systems~\cite{arcile2022timed}, while recent work on hybrid automata focuses on applications to CPSs~\cite{huang2018path, bak2019hybrid}.

% This is especially true for continuous-time systems, where trajectories are partitioned into intervals of discrete actions and the satisfaction of atomic propositions is abstracted over these intervals. This loss complicates the verification of systems under complex specifications like LTL.

% This is particularly true for continuous-time systems, where system state trajectories are divided into intervals of discrete actions, and the satisfaction of atomic propositions is abstracted to entire time intervals. This loss of information complicates the verification and control of systems under complex specifications like LTL.

\subsubsection{Contribution.}
%In this work, 
We introduce a novel approach for abstraction-based system verification for LTL specifications called \emph{\ABZESN{} automata}. 
They encode the abstract properties of \emph{atomic propositions (APs)} along system trajectories through their transition labels, classifying them into four values. 
We define a new construction called \ABZESN{} automaton, which is a translation from an LTL specification for continuous-time systems to a generalized Büchi automaton.
% This construction is inspired by the original one for discrete-time systems~\cite{vardi1986automata}.

% Building on this structure, our verification framework soundly approximates the satisfiability of atomic propositions along system trajectories.
Building on this structure, we propose a verification framework that soundly approximates the satisfiability of atomic propositions along system trajectories.
It relies on a game played by the system and the environment, represented by angelic and demonic nondeterminism, respectively. 
% between two players, representing angelic and demonic types of nondeterminism,
% and an illustrative example for abstraction-based system verification under several LTL specifications. 
% We demonstrate that this construction is sound and applicable for verifying continuous-time systems. 
% We show that this construction is sound and that it can be used for monitoring and verification of continuous-time systems.
Our approach is highly general, supporting nondeterministic, continuous-state, continuous-time systems without global stability assumptions.
To the best of our knowledge, no existing technique provides formal verification for this broad class of systems under LTL specifications. Prior work has instead focused on discrete-time systems~\cite{belta2007symbolic, belta2019formal, banerjee2025challenges}, imposed more restrictive dynamics or assumptions~\cite{tabuada2007approximate, tabuada2008approximate, meyer2019hierarchical, ren2024zonotope}, or addressed smaller classes of specifications~\cite{
%hsu2018lazy, macoveiciuc2019memory, 
pruekprasert2020symbolic, pruekprasert2021fast, Zamani2012, hashimoto20}.
%Timed and 
Hybrid automata~\cite{%alur1994theory,
alur1991hybrid}
%abstract continuous-time systems into discrete-state machines amenable to verification,
can also be used for verification~\cite{bak2019hybrid},
but most of their problems remain undecidable in general.
To achieve this level of generality, we impose a constraint on the satisfaction zones of atomic propositions, a condition met by many systems in practice.

\subsubsection{%Plan
Outline.}
The rest of the paper  
proceeds as follows.
Section~\ref{section:dynamical-systems} introduces  
systems, specifications, and
the verification problem. 
Section~\ref{section:system-abstraction} 
% presents the abstraction of dynamical systems into symbolic models.
shows how to soundly abstract a dynamical system into a finite symbolic model. 
Section~\ref{section:ABZESN} presents \ABZESN{} automata as an abstraction of the specification.
Section~\ref{section:verification} proposes the verification algorithm. 
Section~\ref{section:example} provides an illustrative application example. 
%Finally,
% Section~\ref{section:conclusions} presents the conclusions.
\switchversion{The omitted proofs can be found in 
the appendix~\cite{pruekprasert2025}.}{}

\subsubsection{Notations.}
We write $\R$, $\Rpos$, $\Rnonneg$, $\Z$, and $\Znonneg$ for the sets of real, positive real, nonnegative real, integer, and nonnegative integer numbers. %, respectively.
The infinity norm  is $\infinitynorm{x} = \max_{i=0}^{n-1} \absolutevalue{x_i}$ for $x \in \R^n$. 
Let $2 = \set{\top, \bot}$ be the set of booleans, and $Y^{X}$ the space of functions from $X$ to $Y$. 
We use $X^*$ (\emph{resp.} $X^\omega$) for the set of finite (\emph{resp.} infinite) sequences of elements of $X$.
We use ``iff'' for ``if and only if''.

\section{%Symbolic Controller Synthesis Problem
%\ssc{
Dynamical Systems and LTL Specifications
%}
}
\label{section:dynamical-systems} 
 In this section, we formally
introduce dynamical systems, LTL specifications, and the verification problem for the system.

\subsection{Nondeterministic Dynamical Systems}
We consider dynamical systems $\Sigma = (X, \xi, x_{\ini})$ where
$X \subseteq \R^n$ is the set of considered $n$-dimensional system states, 
$\xi: (2^{X}  \setminus \{\emptyset \}) \times \Rnonneg \to 2^{X} \setminus \{\emptyset \}$ is the system evolution function, and $x_{\ini} \in X$ is the initial state.
For a set $\mathfrak{x} \in 2^{X}\setminus \{\emptyset \}$ of states and a time instant $t \in \Rnonneg$,
the set $\xi(\mathfrak{x}, t) \in  2^{X}  \setminus \{\emptyset \}$ contains all possible states reachable from some state in the set $\mathfrak{x}$ at time 
$t$. 
We require that $\xi$ satisfy the following properties: $\xi(\mathfrak{x},0) = \mathfrak{x}$ and $\xi(\mathfrak{x},t_1 + t_2) = \xi(\xi(\mathfrak{x},t_1), t_2)$.
Note that, unlike symbolic control approaches such as~\cite{pruekprasert2020symbolic, pruekprasert2021fast}, we do not take control signals as inputs to the system’s evolution function. Nevertheless, the results presented in this paper are applicable to controlled systems with fixed control strategies, as these systems can be modeled as dynamical systems $\Sigma$ given above.
% We consider dynamical systems $\Sigma = (X, \UUU, \xi, x_{\ini})$ where 
% $X \subseteq \R^n$ is the set of considered $n$-dimensional system states, $\UUU$ is the set of control signals, %of length $\tau$, 
% $\xi: (2^{\R^n}  \setminus \{\emptyset \}) \times \UUU \to 2^{\R^n} \setminus \{\emptyset \}$ is the system evolution function, and $x_{\ini} \in X$ is the initial state. 
% Let $U$ be the set of all possible control inputs, then a control signal $u \in \UUU$  
% is a function  
% $u: [0, \tau_u] \to U$ 
% that maps each time instant in  
% $[0, \tau_u]$ 
% to a control input in $U$.
% For a set $\mathfrak{x} \in 2^{\R^n}\setminus \{\emptyset \}$ of states and a control signal $u \in \UUU$,
% the set $\xi(\mathfrak{x}, u) \in  2^{\R^n}  \setminus \{\emptyset \}$ contains all possible states reachable from some state in the set $\mathfrak{x}$ at time 
% $\tau_u$. 
% We require that $\xi$ satisfy the following properties: $\xi(\mathfrak{x},\epsilon) = \mathfrak{x}$ and $\xi(\mathfrak{x},u_1 u_2) = \xi(\xi(\mathfrak{x},u_1), u_2)$, where $u_1 u_2(t)= u_1(t)$ if $t <\tau$ and $u_1 u_2(\tau + t)= u_2(t)$ otherwise.
% The system $\Sigma$ is forward-complete \cite{Angeli1999} as
% $\xi(\mathfrak{x}, t)$ cannot be empty.
% Notice also that we define the dynamics on the real coordinate $n$-space $\R^n$, even though we only consider states in the bounded subspace $X$. 
% \ssc{This means that the system state may leave the  
% region $X$ unless controlled under a suitable controller.}\todo{remove?}
For $x \in X$, by abuse of notation, we write $\xi(x, t)$ for $\xi(\{x\}, t)$. 
A trajectory (\emph{resp.} finite trajectory) of $\Sigma$ from $x$ is a function $\sigma \colon \Rnonneg \to X$ 
(\emph{resp.} $\sigma \colon [0, T] \to X$, where $T \in \Rnonneg$)
such that $\sigma(0) = x$ and $\sigma(t) \in \xi(\sigma(0), t)$ for all time $t \in \Rnonneg$ (\emph{resp.} $t \in [0, T]$). Let $\Traj(\Sigma)$ denote all possible (infinite) trajectories of $\Sigma$ starting from $x_\ini$.

% A controller controls the system by issuing control signals based on observed system states. As depicted in Fig.~\ref{fig:controledsystem}, 
% we consider controllers that issue signal 
% $u_k$ at each step $k\in \Zpos$ based on the trajectory 
% $\sigma_1 \ldots \sigma_k$ of previously visited system states, where $\sigma_i$ is the trajectory followed by the system under the signal $u_{i-1}$.
% Formally, a trajectory of $\Sigma$ from $x_0$ along $u$ is a function $\sigma \colon [0, \tau_u) \to X$ such that $\sigma(0) = x_0$ and $\sigma(t) \in \xi(\sigma(0), \restr{u}{t})$, where $\restr{u}{t}$ is the restriction of $u$ to time instants $[0, t)$.
% A trajectory $\sigma_1 \ldots \sigma_k \ldots$ is induced by a strategy $\pi$ if $\sigma_1(0) = x_{\ini}$ and, for all $k \in \Zpos$, $\sigma_k$ is a trajectory of $\Sigma$ from $\sigma_{k-1}(\tau)$ along $\pi(\sigma_1 \ldots \sigma_{k-1})$.
% We denote by $\Traj(\pi)$ the set of trajectories induced by $\pi$.
% %\ssc{
% For ease of presentation, we assume that the controller only issues signals of length 
% $\tau$, resulting in a periodic control scheme. However, the presented results can be extended to  other control schemes.

\subsection{Atomic Propositions and Assumptions on Trajectories}
\label{section:AP}
%\ssc{
Atomic propositions (AP), statements about a state of the system, are the basic building blocks of temporal logic formulas for specifications in this paper.
%}
Examples of atomic propositions include properties such as whether the system is colliding with an obstacle or whether its position is in a desirable region. %, or whether the battery's level is critical.
Let $\AP$ denote the finite set of considered atomic propositions, and $P: X \to 2^{\AP}$ represent the set of atomic propositions that hold at each system state:
if $a \in \AP$ represents the property that the system is safe, 
then $P(x)(a) = \top$ means that the system is safe at state $x \in X$.  In other words, $\setcomp{x \in X}{P(x)(p) = \top}$ is the region of states that satisfies the atomic proposition $p \in \AP$.

\subsection{Classic LTL with Signal Semantics}\label{section:classicLTL}
% \todo[inline]{usual discrete-time semantics -  our previous work tries to consider cont.-time semantics -  that was why we consider subclasses of LTL -  Here comes this work!}
% \todo[inline]{type signals (especially the space)$\rightarrow$ Section2 ?}

% \begin{defn}%[LTL formulas]
%     An \emph{LTL formula} on a set $\AP$ of atomic propositions is an expression generated by the following grammar:
%     \[
%         \phi \Coloneqq \top \gramor p \gramor \neg \phi \gramor \phi \lor \phi \gramor \Next \phi \gramor \phi \Until \phi,
%     \]  
%     where $\top$ denotes \emph{true} and $p \in \AP$ is an atomic proposition.
%     %(see Section~\ref{section:AP}).
% \end{defn}

Linear Temporal Logic  (LTL) formulas are 
generated by the following grammar:
\[
    \phi \Coloneqq \top \gramor p \gramor \neg \phi \gramor \phi \lor \phi \gramor \Next \phi \gramor \phi \Until \phi,
\]  
where $\top$ is syntax for \emph{truth} (not to be confused with the semantic boolean $\top \in 2$) and $p \in \AP$ is an atomic proposition.

Conventionally, the semantics of LTL is defined on words, i.e., in discrete time.
For example, a classic LTL formula may contain $\Next \phi$ (\emph{next} $\phi$), which holds for $x_i x_{i+1} \ldots$ if
$\phi$ holds at the \emph{next} discrete step $x_{i+1} x_{i+2} \ldots$ (see \cite{pnueli1977temporal} for a formal definition).
% Conventionally, semantics of an LTL formula is defined for a sequence $p_0 p_1 \ldots \in (2^{\AP})^\omega$, and is therefore suitable to represent a property of a discrete run such as $x_0 u_0  x_1 \ldots \in X (\UUU  X)^\omega$ where $P(x_i) = p_i$ for all $i \in \N$. 
However, we are interested in the property of a system trajectory $\sigma: \Rnonneg \to X$ defined on the continuous timeline. 
We consider \emph{AP-signal} $\varsigma \colon \Rnonneg \to 2^{\AP}$ where $\varsigma(t) = P(\sigma(t))$, i.e., $\varsigma$ indicates the atomic propositions that hold along $\sigma$. 
Note that trajectories $\sigma \colon \Rnonneg \to X$ and AP-signals $\varsigma \colon \Rnonneg \to 2^{\AP}$ have slightly different types.
We say that a formula is \emph{continuous-time} if it contains no subformulas of the form $\Next \phi$. 
Then, 
the \emph{signal semantics} of continuous-time LTL is the relation $\vDash$ defined on $\varsigma$ as follows:

\begin{itemize}
\setlength{\multicolsep}{0pt}
\begin{multicols}{2}
    \item $\varsigma, t \vDash \top$ always,
    \item $\varsigma, t \vDash p$ iff 
    $\varsigma(t)(p) = \top$,
    \item $\varsigma, t \vDash \neg \phi$ iff $\varsigma, t \vDash \phi$ does not hold,
    \item $\varsigma, t \vDash \phi \lor \psi$ iff $\varsigma, t \vDash \phi$ or $\varsigma, t \vDash \psi$,
\end{multicols}
    \item $\varsigma, t \vDash \phi \Until \psi$ iff
    %there exists
    $\exists t' \geq t$ such that $\varsigma, t' \vDash \psi$ and for all $t'' \in [t,t')$, $\varsigma, t'' \vDash \phi$.
    % \item $\varsigma, t \vDash \phi \Until \psi$ iff
    % %there exists $t' \geq t$ such that $\varsigma, t' \vDash \psi$ and for all $t'' \in [t,t')$, $\varsigma, t'' \vDash \phi$.
\end{itemize}

% \blind{
% \begin{itemize}
%     \item $\varsigma, t \vDash \top$ always,
%     \item $\varsigma, t \vDash p$ iff 
%     $\varsigma(t)(p) = \top$,
%     \item $\varsigma, t \vDash \neg \phi$ iff $\varsigma, t \vDash \phi$ does not hold,
%     \item $\varsigma, t \vDash \phi \lor \psi$ iff $\varsigma, t \vDash \phi$ or $\varsigma, t \vDash \psi$,
%     \item $\varsigma, t \vDash \phi \Until \psi$ iff there exists $t' \geq t$ such that $\varsigma, t' \vDash \psi$ and for all $t'' \in [t,t')$, $\varsigma, t'' \vDash \phi$.
% \end{itemize}}

% \begin{defn}%[Signal semantics]
%     The \emph{signal semantics} of continuous-time LTL is a relation $\vDash$ between signals $\sigma \colon \Rnonneg \to \R^n$, time instants $t \in \Rnonneg$ and continuous-time LTL formulas $\phi$ defined
%     %by induction on $\phi$
%     as follows:
%     \begin{itemize}
%         \item $\sigma, t \vDash \top$ always,
%         \item $\sigma, t \vDash p$ iff $p(\sigma(t))$ holds,
%         \item $\sigma, t \vDash \neg \phi$ iff $\sigma, t \vDash \phi$ does not hold,
%         \item $\sigma, t \vDash \phi \lor \psi$ iff $\sigma, t \vDash \phi$ or $\sigma, t \vDash \psi$,
%         \item $\sigma, t \vDash \phi \Until \psi$ iff there exists $t' \geq t$ such that $\sigma, t' \vDash \psi$ and for all $t'' \in [t,t')$, $\sigma, t'' \vDash \phi$.
%     \end{itemize}
% \end{defn}

% \todo[inline]{Introduce release}
% \begin{align*}
%     \Finally \phi &= \top \Until \phi \\
%     \Globally \phi &= \neg \Finally \neg \phi \\
%     \phi \Release \psi &= \neg (\neg \phi \Until \neg\psi).
% \end{align*} 
The formula
$\phi \Until \psi$ ($\phi$ \emph{until} $\psi$) means that $\phi$ must remain true until $\psi$ becomes true, and $\psi$ must become true at some point.
We also use the usual shorthands: 
$\phi \land \psi = \neg (\neg \phi \lor \neg \psi)$,
% $\phi \to \psi = \neg \phi \lor \psi$,
$\phi \Release \psi = \neg (\neg \phi \Until \neg\psi)$,
$\Finally \phi = \top \Until \phi$, and $\Globally \phi = \neg \Finally \neg \phi$. 
%On the other hand,
The formula
$\phi \Release \psi$ ($\phi$ \emph{release} $\psi$) means that $\psi$  must remain true until $\phi$ becomes true, and $\psi$ must remain true forever if $\phi$ never becomes true.
%We also introduce the following shorthand:
%\ssc{
The formula $\Finally \phi$ (\emph{eventually} $\phi$) means that $\phi$ will hold at some point, while $\Globally \phi$ (\emph{globally} $\phi$) means that $\phi$ holds all the time.
%
%This makes LTL a very expressive logic.
%\ssc{
LTL is a very expressive logic.
%}
For example, a reach-avoid specification can be represented as
$
    \Globally \neg \text{a} \land \Finally \text{r},
$
where $\text{a}$ is an atomic proposition that holds on the zone to avoid and $\text{r}$ is the one that holds on the zone to reach.
% For this reason, LTL is a formalism widely used 
% %for describing control and verification specifications
% for describing specifications~\cite{meyer2019hierarchical, pruekprasert2021fast, ren2024zonotope, pruekprasert2020symbolic, belta2019formal, belta2007symbolic, pnueli2005applications, chai2014online}.

\subsection{
%\ssc{
System Verification for LTL Specifications
%}
}\label{section:System Verification for LTL specifications}
%Problem Formulation}

% We are interested in verification.
%Monitoring is the fact of checking whether an execution of the system satisfies a given specification, while verification checks whether a control strategy only produces executions that satisfy that specification.

%\ssc{ 
This work considers system verification under LTL specifications, i.e., checking whether the system only produces trajectories that satisfy a given specification.
Formally, we consider the following problem.
%}

\begin{problem}\label{problem:continuousVerification}
Given system $\Sigma$, $P \colon X \to 2^{\AP}$, %, %a control strategy $\pi$, 
and a continuous-time LTL specification $\phi$, our goal is to verify  
whether $P \circ \sigma, 0 \vDash \phi$ for all $\sigma \in \Traj(\Sigma)$.
%  whether for all $\sigma \in \Traj(\Sigma)$, %\Traj(\pi)$, 
% $P \circ \sigma, 0 \vDash \phi$.
% % , where $\varsigma(t) = P(\sigma(t))$.
\end{problem}

%\ssc{ 
A standard approach to verification is to construct
% the product of the controlled system's automaton with
a Büchi automaton corresponding to the LTL formula, as
%} 
it is well-known~\cite{vardi1986automata, gastin2001fast} that LTL formulas can be translated to Büchi automata in the following sense.

% We are interested in verification, i.e., checking whether a control strategy only produces executions that satisfy a given specification.
% Formally, given a system $\Sigma$, $P \colon X \to 2^{\AP}$, a control strategy $\pi$, and a specification $\phi$, verification is an algorithm that checks whether for all $\sigma \in \Traj(\pi)$, $\sigma, 0 \vDash \phi$.
 % \clovis{Why are we interested in this? -> it solves verification}
% It is well-known\todo{citation} that LTL formulas can be translated to Büchi automata in the following sense:
\begin{proposition}[{\cite[Theorem~2.1]{vardi1986automata}}]\label{prop:LTLtoBuchi}
    For all LTL formulas $\phi$, there exists a Büchi automaton $\BBB$ such that for all words $w \colon \Znonneg \to 2^{\AP}$, $w,0 \vDash \phi$ iff $w \in \LLL(\BBB)$.
\end{proposition}
% \clovis{Problem: we have not introduced LTL word semantics\todo{(Sasinee) probably OK, NO?}}   
% \ssc{We invite interested readers to refer to}~\cite{vardi1986automata, gastin2001fast} \ssc{for the translation algorithm of Proposition}~\ref{prop:LTLtoBuchi}. \ssc{However, we briefly explain the key concepts here.}
We refer interested readers to~\cite{vardi1986automata, gastin2001fast} for the translation algorithm of Proposition~\ref{prop:LTLtoBuchi}. However, we briefly explain the key concepts here.
The translation
heavily relies on the fact that $\phi \Until \psi$ is equivalent to $\psi \lor (\phi \land \Next (\phi \Until \psi))$, and similarly $\phi \Release \psi$ is equivalent to $\psi \land (\phi \lor \Next (\phi \Release \psi))$.
% Using this fact, one can build from $\phi$ a generalized Büchi automaton whose action labels are atomic propositions and states are functions that map each subformula of $\phi$ to true or false, and uses the equivalences above to decide the value of formulas of the form Until or Release.
% For instance, let us consider the formula $p \Until p'$, where its subformulas are $p\in \AP$, $p'\in \AP$ and $p \Until p'$.
% If $p$ is mapped to false, and $p'$ and $p \Until p'$ to true in a state $x_i$, then this state can only transition to a state $x_{i+1}$ where $p \Until p'$ is mapped to true.
% This transition can be interpreted as follows: if $p \Until p'$ is true but $p$ is false for $x_i x_{i+1} \ldots$, then $p \Until p'$ must be true for the next-step sequence $x_{i+1} x_{i+2} \ldots$. 
% The initial state is a dummy state with transitions to all states that maps the main formula ($p \Until p'$ in this case) to true.
% \todo[inline]{This is wrong. rewrite. or do we have to explain this much?}
For example, if $p$ and $p \Until p'$ hold at time $k$, but $p'$ does not, then necessarily $p \Until p'$ must hold at time $k+1$.
Using this fact, it is possible to build a generalized Büchi automaton for $\phi$ whose action labels are valuations of atomic propositions and whose states are valuations of subformulas of $\phi$.
The Büchi automaton's transitions reflect behaviors as described above: in a state where $p$ and $p \Until p'$ hold but $p'$ does not, it can only transition to a state where $p \Until p'$ holds.
Its accepting sets ensure that if $\phi_1 \Until \phi_2$ holds at some point, then $\phi_2$ must hold at some later point.
The accepting states are the states that capture a property of the operator $\Until$ that cannot be verified by comparing two consecutive states in a run. In this case, in the accepting states, either $p'$ holds or $p \Until p'$ does not, due to the fact that if $p \Until p'$ holds in some state $x_i$, then eventually $p'$ later holds at some state $x_j$ where $j\geq i$.
% \todo[inline]{Something about accepting states}
%For more details, see... \todo{citation}

\section{System Abstraction and Information Loss}
\label{section:system-abstraction}

%\todo[inline]{abstraction: time, space $\to$ finite state space}

% \ssc{
% A dynamical system as described in the previous section is a continuous-state, continuous-time system.
% In order to verify a system under an LTL specifciation using the product-automaton approach  in} Section~\ref{section:System Verification for LTL specifications}, 
% \ssc{we need to abstract the system into a finite-state discrete-time \emph{symbolic model} that approximates the behavior of the dynamical system.
% %retains significantly
% %less information than their original systems. 
% }
A dynamical system, as described in the previous section, is a continuous-state, continuous-time system.
In order to verify a system under an LTL specification by checking the system with the corresponding Büchi automaton, we need to abstract the system into a finite-state, discrete-time \emph{symbolic model} that approximates the behavior of the dynamical system.
%retains significantly
%less information than their original systems. 

\subsection{Time-abstraction and
Signal Chopping
\label{section:time-abstraction}
}
%LTL Signal Semantics}

When abstracting a system, one of the most important losses of information comes from discretizing time. 
Indeed, since we are interested in complex 
temporal specifications, where the order in which atomic propositions are satisfied matters, discretizing time loses information about whether an LTL formula holds between two time instants.
This information can be arbitrarily complex, and any abstraction into a finite number of patterns necessarily induces a loss of precision.

% We discretize time by chopping an AP-signal $\varsigma \colon \Rnonneg \to 2^{\AP}$ into slices of a fixed length  $\tau \in \Rpos$ (see Section~\ref{section:ABZESN} for a formal definition). 
% Here, we focus on $\widebar{\varsigma}  \colon [0, \tau] \to 2^{\AP}$ that is a slice of length
% $\tau$ taken from $\varsigma$.
% We abstract the satisfaction of each $p \in \AP$ along 
% $\widebar{\varsigma}$ into one of four possible patterns $\OOO = \{A, Z, E, N\}$, referred to as an \emph{observation}. 
% We define a \emph{AP-valuation} function $o_{\widebar{\varsigma}} \colon \AP \rightharpoonup \OOO$, which assigns observations  along $\widebar{\varsigma}$. For each $p \in \AP$,
% \begin{enumerate}
%     \item $o_{\widebar{\varsigma}} (p) = A$ if  
%     $\widebar{\varsigma}, t \vDash p$ holds for all $t \in [0, \tau)$.
%     \item $o_{\widebar{\varsigma}}(p) = Z$ if there exists $t^* \in [0, \tau)$ such that $\widebar{\varsigma}, t \vDash p$ holds for all $t \leq t^*$ and $\widebar{\varsigma}, t \vDash p$ does not hold for all $t > t^*$.
%     \item $o_{\widebar{\varsigma}}(p) = E$ if there exists $t^* \in [0, \tau)$ such that $\widebar{\varsigma}, t \vDash p$ does not hold for all $t \leq t^*$ and $\widebar{\varsigma}, t \vDash p$ holds for all $t > t^*$.
%     \item $o_{\widebar{\varsigma}}(p) = N$ if $\widebar{\varsigma}, t \vDash \psi$ does not hold for all $t \in [0, \tau)$.
%     \item $o_{\widebar{\varsigma}}(p)$ is undefined if $o_{\widebar{\varsigma}}(p) \notin \{A, Z, E, N\}$.
% \end{enumerate}

We discretize time by chopping an AP-signal $\varsigma \colon \Rnonneg \to 2^{\AP}$ into slices of a fixed length  $\tau \in \Rpos$. 
We abstract the satisfaction of each $p \in \AP$ within each of these slices into one of four possible patterns $\OOO = \{A, Z, E, N\}$, referred to as an \emph{observation}. 
Conceptually, $A$ means that the $p$ holds at \textbf{A}ll time throughout the interval, $Z$ means that $p$ holds only at the beginning of the interval (time \textbf{Z}ero), $E$ means that $p$ holds only at the \textbf{E}nd of the interval, and $N$ means that $p$ holds at \textbf{N}one of the interval time points.
Formally, we define the \emph{signal chopping} of $\varsigma$ along $\tau$, denoted $\chopping{\varsigma}{\tau} \colon \N \to \OOO^{\AP}$, as follows: for all $n \in \N$,
\begin{itemize}
    \item $\chopping{\varsigma}{\tau}(n)(p) = A$ if for all $t \in [n\tau,(n+1)\tau]$, $\varsigma(t)(p) = \top$, 
    \item $\chopping{\varsigma}{\tau}(n)(p) = Z$ if there exists $t' \in [n \tau, (n+1)\tau)$ such that $\varsigma(t)(p) = \top$ for all $t \leq t'$ and $\varsigma(t)(p) = \bot$ for all $t > t'$.
    \item $\chopping{\varsigma}{\tau}(n)(p) = E$ if there exists $t' \in [n \tau, (n+1)\tau)$ such that $\varsigma(t)(p) = \bot$ for all $t \leq t'$ and $\varsigma(t)(p) = \top$ for all $t > t'$. 
    \item $\chopping{\varsigma}{\tau}(n)(p) = N$  if for all $t \in [n\tau,(n+1)\tau]$, $\varsigma(t)(p) = \bot$. 
    % \item $\chopping{\varsigma}{\tau}(n)(p)$ is undefined if $\chopping{\varsigma}{\tau}(n)(p) \notin \{A, Z, E, N\}$.
\end{itemize}

The slice $\chopping{\varsigma}{\tau}(n)(p)$ is undefined if $\chopping{\varsigma}{\tau}(n)(p) \notin \{A, Z, E, N\}$.
To ensure that $\chopping{\varsigma}{\tau}(n)$ is well-defined for all $p \in \AP$, we impose the following assumptions.

\begin{asm}\label{asm: restriction BSZE}
    For all trajectories $\sigma$, all $p \in \AP$, all $t \in \Rnonneg$, and all $t' \in [0, \tau]$,
    \begin{align}
  \label{eq: no B,S}
  \begin{split}
    P(\sigma(t))(p) = P(\sigma(t + t'))(p) \Rightarrow
    \forall t'' \in [0, t'],\ 
    P(\sigma(t))(p) = P(\sigma(t + t''))(p) .
  \end{split}
  \\
  \label{eq: Z E}
  \begin{split}
    P(\sigma(t))(p) \neq P(\sigma(t + \tau))(p) \text{ and } 
    P(\sigma(t))(p') \neq P(\sigma(t + \tau))(p') 
    \Rightarrow p = p'.
  \end{split}
\end{align}
\end{asm}

The property in \eqref{eq: no B,S} restricts that, within time $\tau$, a system trajectory cannot cross the border of each AP region twice.
It is possible to ensure that the system trajectories have this property by appropriately designing or selecting a Lyapunov-like barrier function (see, e.g., \cite{lygeros2003dynamical, panagou2015distributed}) to enforce that any deviation of $\sigma(t)$ from the initial AP region results in a monotonic decrease in a certificate function over $[0, \tau]$, thus preventing the system from returning to its initial AP region within the time horizon.
The property in \eqref{eq: Z E} implies that a system trajectory can cross at most one AP region boundary within a time interval of length $\tau$.
To enforce this property, one may take $\tau$ small enough so that the minimum distance between the boundaries of any two AP regions is greater than the distance the system can travel in time $\tau$.
This is made formal by the following lemma, whose proof is provided in~\switchversion{\cite[Appendix~A.1]{pruekprasert2025}}{Appendix~\ref{proof:lem:bounded speed}}.

%~\cite[Appendix~\switchversion{A.1}{\ref{proof:lem:bounded speed}}]{pruekprasert2025}.

% i.e., an upper limit of the distance that the system can travel in time $\tau$.
\begin{lemma}\label{lem:bounded speed}
    System $\Sigma$ has \emph{bounded speed} %, i.e., 
    if there exists $\Delta \colon \Rpos \to \Rpos$ such that for all $t \in \Rpos$, $x \in X$, $t' \leq t$, and $y \in \xi(x, t')$, $\infinitynorm{y-x} \leq \Delta(t)$.
    It is \emph{AP-separated} if for all $p \neq p' \in \AP$, there exists $d_{p,p'} > 0$ such that the distance from the (topological) boundary of $\setcomp{x \in X}{p \in P(x)}$ and $\setcomp{x \in X}{p' \in P(x)}$ is at least $d_{p,p'}$.
    If $\Sigma$ is speed-bounded and AP-separated, then for any choice of $\tau \leq \inf_{p \in \AP} \inf_{p' \in \AP, p' \neq p} \inf \Delta^{-1}(d_{p,p'})$,~\eqref{eq: Z E} holds.
\end{lemma} 
The two properties in Assumption~\ref{asm: restriction BSZE} are necessary because we want to prevent different subformulas from changing truth value at different times in the same time interval of length $\tau$ and also to prevent two subformulas from having $Z$ and $E$ as observations during the same time interval of length $\tau$. 
Otherwise, we need to introduce new observations: $B$ (if a formula holds at \textbf{B}oth ends of the interval, but not on the whole interval) and $S$ (if it holds \textbf{S}omewhere but not at the ends).
In this paper, we show that it is possible to deduce the observation of all subformulas from those of atomic propositions (see Lemma~\ref{lem:ap-word-extension-accepting-run}).
If we allow these new observations, the result no longer holds, and it is unclear how to construct a sound translation.
From Assumption~\ref{asm: restriction BSZE}, we get the following lemma.
\begin{lemma}\label{lem: BSZN}
    Assume both properties in Assumption~\ref{asm: restriction BSZE}. 
    Given a trajectory $\sigma$, let
    $\chopping{\varsigma}{\tau} \colon \N \to {\OOO}^{\AP}$ be the chopped AP-signal of $\varsigma = P\circ \sigma$.
    For all $n \in \N$,
     \begin{enumerate}
        \item For all $p \in \AP$, we have its observation $\chopping{\varsigma}{\tau}(n)(p) \in \{A, Z, E, N\}$. 
        \item For all $p, p' \in \AP$, if $\chopping{\varsigma}{\tau}(n)(p) \in \{Z, E\}$ and $\chopping{\varsigma}{\tau}(n)(p') \in \{Z,E\}$, then $p = p'$.
    \end{enumerate} 
    % $\widebar{\varsigma}  \colon [0, \tau] \to 2^{\AP}$ be a slice of length $\tau$ taken from the AP-signal $\varsigma$ of $\sigma$. 
    % \begin{enumerate}
    %     \item For all $p \in \AP$, we have its observation $o_{\widebar{\varsigma}}(p) \in \{A, Z, E, N\}$. 
    %     \item For all $p, p' \in \AP$, $o_{\widebar{\varsigma}}(p') \in \{Z, E\}$ and $o_{\widebar{\varsigma}}(p') \in \{Z,E\}$ implies $p = p'$.
    % \end{enumerate} 
\end{lemma}
% \begin{proof}
%     The property 2 holds by \eqref{eq: Z E}. The property 1 holds by \eqref{eq: no B,S} and \eqref{eq: Z E}.
  
% \end{proof}

Lemma~\ref{lem: BSZN} indicates that at any time step $n \in \N$: $\mathit{1.}$
$\chopping{\varsigma}{\tau}(n) \in \OOO$ for all $p \in \AP$, and $\mathit{2.}$ the system trajectory crosses at most one AP-region border during the time interval $[n \tau, (n+1) \tau]$.
By this lemma and the definition of signal slicing, we also have the following corollary. 
\begin{corollary}
    \label{cor:signal-chopping-is-ap-word}
    Assume both properties in Assumption~\ref{asm: restriction BSZE}. 
    Given a %system 
    trajectory $\sigma$, 
    let
    $\chopping{\varsigma}{\tau} %\colon \N \to \OOO^{\AP}
    $ be the chopped AP-signal of $\varsigma = P\circ \sigma %\colon \Rnonneg \to 2^{\AP}
    $.
    For all $n \in \N$ and $p \in \AP$,
    \begin{equation*} 
    \chopping{\varsigma}{\tau}(n)(p) \in \{A, E\} \iff \chopping{\varsigma}{\tau}(n+1)(p) \in \{A, Z\}.
\end{equation*}
\end{corollary}

Corollary~\ref{cor:signal-chopping-is-ap-word} follows from Lemma~\ref{lem: BSZN}.
Indeed, because we use closed intervals, an atomic proposition holds at the end of an interval ($A$ or $E$) iff it holds at the beginning of the next one ($A$ or $Z$).

% Corollary~\ref{cor:accepting-word-signal} follows from Lemma~\ref{lem: BSZN}, and because partitioning the signals into slices over closed intervals guarantees that if an atomic proposition holds at the end of one interval ($A$ or $E$), it also holds at the start of the next ($A$ or $Z$).
% \todo[inline]{ remark that it works because we consider closed intervals.}

\subsection{Symbolic Models}\label{section:Symbolic Models and Control Strategies}

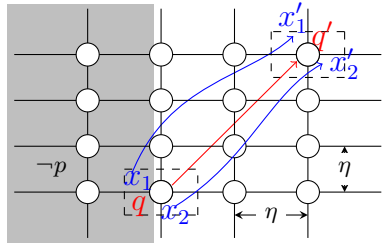
\begin{figure}[t]
% \resizebox{0.35\textwidth}{0.35\textwidth}
% \resizebox{0.4\textwidth}{0.2\textwidth}
{
\begin{tikzpicture}[x=0.08\textwidth, y=0.05\textwidth]
\tikzset{dot/.style={fill=white, draw=black,circle}}
\node[dot,fill=red] at (2,1){}; 
\draw[dashed] (1.5,0.5) -- (2.5,0.5);
\draw[dashed] (1.5,1.5) -- (2.5,1.5);
\draw[dashed] (1.5,0.5) -- (1.5,1.5);
\draw[dashed] (2.5,0.5) -- (2.5,1.5);
\node[dot,fill=red] at (4,4){};
\draw[dashed] (4.5,4.5) -- (3.5,4.5);
\draw[dashed] (4.5,3.5) -- (3.5,3.5);
\draw[dashed] (4.5,4.5) -- (4.5,3.5);
\draw[dashed] (3.5,4.5) -- (3.5,3.5);
% \path [->,draw=blue, snake it]
%     (1.7,1.4) -- (1.7,3.4) -- (2.0,3.5) -- (3, 3.2)--(3.5, 4.0)--(3.8, 4.4);
\path[->, draw=blue]
  (1.6,1.4) .. controls (2.1,3.4) and (3.2,3.5) .. (3.8, 4.4); 
% \node [blue, font = \large] at (1.4,2.2) {$u$}; 
% \node [blue, font = \large] at (3.5,1.5) {$u$}; 

\draw[->, red] plot[smooth,tension=0.3] coordinates
    {(2.15,1.15)   (3.85,3.85)};
\node [red, font = \large] at (1.7,0.7) {$q$}; 
% \node [red, font = \large] at (2.7,2.5) {$u$}; 

\draw [stealth-stealth](4.5,1) -- (4.5,2);
\node [fill=white] at (4.5, 1.5) 
    {$\eta$};

\draw [stealth-stealth](3,0.5) -- (4,0.5);
\node [fill=white] at (3.5, 0.5) 
    {$\eta$};

\begin{scope}[on background layer]   
          % \filldraw[lightgray,line width=50pt] (1,1) --  (1,2) -- cycle;
        \fill[lightgray] (-0.1,-0.1) rectangle (1.9,5.1); 
        \end{scope}
\node [fill=lightgray] at (0.5, 1.5) 
    {$\neg p$};
    
% changed by clovis
\node [red, font = \large] at (4.2,4.4) {$q'$}; 
\foreach\l[count=\y] in {A,...,D}
{
    \draw (0.0,\y) -- (5,\y);
    % \node at (-0.5,\y){\l};
}
\foreach \x in {1,2,...,4}{
    \draw (\x,0) -- (\x,5);
    % \node at (\x,-0.5){\x};
        \foreach \y in {1,...,4}
            \node[dot] at (\x,\y){};
}
% \node[dot,pattern color=teal, pattern={north east lines}] at (3,1){};
% \node[rectangle,fill=teal, inner sep=1.7pt,minimum size=2pt] at (2.5,0.7){};
% \node [teal, font = \large] at (3.3,0.7) {$q''$};
% changed by clovis
\node [blue, font = \large] at (1.7,1.25) {$x_1$}; 
\node [blue, font = \large] at (3.8, 4.8) {$x'_1$ };
\node [blue, font = \large] at (2.2,0.5) {$x_2$}; 
\node [blue, font = \large] at (4.5, 3.8) {$x'_2$ }; 
% added by clovis
% \path [->,draw=blue, snake it]
%     (2.2,0.7) -- (3.2, 1.7) -- (3.5, 3.2) -- (4.2, 3.8);
\path[->, draw=blue]
  (2.2,0.7) .. controls (3.2, 1.7) and (3.5, 3.2) .. (4.2, 3.8);
% added by clovis
\end{tikzpicture}
}
\centering
\caption{A quantized state space using a quantization parameter $\eta \in \Rpos$. 
The circles symbolize discrete states.
Each discrete state represents (e.g., $q$ and $q'$) a corresponding state region of the size
$\eta \times \eta$ (e.g., the dashed boxes around $q$ and $q'$ circles, respectively). The atomic proposition $p$ holds 
 at all states, except those in the left gray half space.
\label{fig:quantized}
}
\end{figure}

We consider a symbolic model that serves as an abstraction of the dynamical system $\Sigma$, not only with respect to the time interval $\tau$ discussed in the previous section, but also by abstracting the continuous state space into a discrete set of states.
A symbolic model is a labeled transition system~\cite{baier2008principles} $\SSS = (Q, \delta, q_\ini)$, where $Q$ is a discrete set of states, $\delta \subseteq Q \times \OOO^{\AP} \times Q$ is
a transition relation, and $q_\ini \in Q$ is the initial state.
% \ssc{Notice that our symbolic model slightly differs from the usual one as we also consider the $\mathit{AP}$-observations along the transition.}
% Notice that our symbolic model slightly differs from the usual one as transitions are labeled by $\mathit{AP}$-observations, not control signals. 
If a transition $(q, o, q') \in \delta$ exists,
this means that the system may move from state $q$ to $q'$ in exactly time $\tau$, provided that the observations of atomic propositions are those given by the function $o \colon \AP \to \OOO$.
%satisfies $p \in o(p)$. 
The transition system is nondeterministic in the sense that there may exist two transitions $(q, o, q'), (q, o, q'') \in \delta$ with $q' \neq q''$, meaning that the system may transition from $q$ to $q'$ or $q''$. Moreover, there may exist two transitions $(q, o, q'), (q, o', q') \in \delta$ with $o \neq o'$, meaning that the observation of an atomic proposition $p\in \AP$ may be $o(p)$ or $o'(p)$.

% The notions of 
% infinite and finite runs carry directly to the symbolic model.
Symbolic models have discrete executions defined in terms of runs, whereas dynamical systems have continuous executions defined in terms of trajectories.
Namely, an \emph{infinite (\textnormal{resp.\ } finite) run} % (\emph{resp. finite run})
of the symbolic model $\SSS$ is a sequence 
$r_s = q_0  o_0 q_1 \ldots \in Q (\OOO^{\AP} Q)^\omega$
(\emph{resp.} $r_s = q_0 o_0 q_1 \ldots o_{n-1} q_n \in Q (\OOO^{\AP} Q)^*$ with $n \in \N$)
such that %there exists  
$(q_i, o_i, q_{i+1}) \in \delta$ for all $i \in \Znonneg$ (\emph{resp.} $i \in \{0, \ldots, n-1\}$). In what follows, we refer to infinite runs simply as \emph{runs}, and specify \emph{finite runs} explicitly when needed.
%$O_i \in \PPP(\OOO)^{\AP}$ such that $o_i \in O_i$ and
% $(q_i, O_i, q_{i+1}) \in \delta$ for all $i \in \Rnonneg$ (\emph{resp.} $i \in \{0, \ldots, n-1\}$).\todo{why not just $o_i \in \OOO^{AP}$? Should we use $o_i$ or $\nu_i$?}

% Intuitively, a symbolic controller controls the symbolic model $\SSS$ in the same way as a controller controls $\Sigma$, but taking runs as inputs instead of strategies.
% A \emph{symbolic control strategy} is a function $\pi_s: Q \times (\VVV \times \OOO^{\AP} \times Q)^* \to \VVV$.

% A strategy is \emph{memory-finite} if it can be expressed as a function $\pi' \colon Q \times \OOO^{\AP} \times N \to \VVV \times N$, where $N$ is a finite set that represents the memory of the strategy (we can recover a strategy in the usual sense from $\pi'$ by iterating it on the elements of the sequence to update the memory).
% If $\pi'$ is a memory-finite strategy, then the symbolic model $\SSS$ controlled by $\pi'$ can be written as an automaton whose set of states is $Q \times N$ and transitions are given by $\pi'$.

% \subsection{Soundness of Symbolic Models}

A symbolic state $q \in Q$ 
%represents several or 
may represent an infinite number of actual system states.
We are interested in when a symbolic model soundly represents a dynamical system, i.e., when all behaviors of the dynamical system are modeled by those of the symbolic model.
In Fig.~\ref{fig:quantized}, state $q$ represents all states in the dashed square centered around it, including $x_1$ and $x_2$.
In this example, both trajectories from $x_1$ to $x_1'$ and $x_2$ to $x_2'$ are abstracted to the transition from $q$ to $q'$.
However, the atomic proposition $p$ holds throughout the entire trajectory from $x_2$, but not at the beginning of the trajectory from $x_1$.
Therefore, there must exist two transitions, $(q, o, q')$ and $(q, o', q')$, where $o(p) = E$ and $o'(p) = A$, reflecting the fact that the atomic proposition $p$ may hold either only at the end or throughout the entire trajectory.

We consider symbolic models constructed by any method as long as they provide the following information.
\begin{enumerate}
    \item An abstraction map $\abs: X \to Q$ that maps each system state to its corresponding symbolic state.
    %, \ssc{with $\abs^{-1}: Q \to 2^\mathit{X} \setminus \emptyset$}\todo{this is to ensure that each $q \in Q$ has at least one corresponding state in $X$; therefore, a run corresponds to at least one trajectory.}.
    % and the inverse map $\abs^{-1}: Q \to 2^\mathit{X}\setminus \emptyset$.
    For the example in Fig.~\ref{fig:quantized}, the states in the dashed boxes around $q$ and $q'$ are mapped to $q$ and $q'$, respectively.
    % Remark that all states on the grid lines can be mapped to any discrete state within the distance $\eta/2$. For example, the teal square point cap be mapped to either $q$ or $q''$. 
    
    \item It must be so that,  
    % for all $x \in X$ and all $x' \in \xi(x,\tau)$, %and trajectory $\sigma$ from $x$ to $x'$, 
    % there exists the transition $(\abs(x),  o, \abs(x'))\in \delta$ in the symbolic model representing a system trajectory from $x$ to $x'$. Formally,
    \begin{equation}\label{eq: approximate simulation}
        \text{For all }x \in X \text{ and all }x' \in \xi(x,\tau), \text{ there exists } (\abs(x),  o, \abs(x'))\in \delta.
    \end{equation}
    Namely, there always exists a transition $(\abs(x),  o, \abs(x'))\in \delta$ representing a trajectory from $x$ to $x'$.
    This property is known as \emph{approximate simulation} and can be ensured by constructing symbolic models using the methods %proposed 
    in~\cite{Zamani2012, hashimoto20, pruekprasert2020symbolic, pruekprasert2021fast}. 
    % In Fig.~\ref{fig:quantized}, the transition $(q, o, q)$ represents several actual system transitions, including those from $x_1$ to $x'_1$ and from $x_2$ to $x'_2$.
    Note that a transition 
    $(q, o, q')\in \delta$ may represent %several
    infinitely many
    trajectories from $ \abs^{-1}(q)$ to $\abs^{-1}(q')$ under observations given by $o$.
    
    \item Functions $%\rho_A,
    \rho_Z, \rho_E: Q \times Q \times \AP \to \{+, -, ?\}$ 
    from which we define $P_Z$, and $P_E$ as follows:
    \begin{align*}
        P_Z(q, q', p) &=
        \begin{cases}
            \set{A, Z} & \text{if $\rho_Z(q, q', p) = +$} \\
            \set{E, N} & \text{if $\rho_Z(q, q', p) = -$} \\
            \OOO & \text{otherwise}
        \end{cases} \\
        P_E(q, q', p) &=
        \begin{cases}
            \set{A, E} & \text{if $\rho_E(q, q', p) = +$} \\
            \set{Z, N} & \text{if $\rho_E(q, q', p) = -$} \\
            \OOO & \text{otherwise.}
        \end{cases}
    \end{align*}
    They must be such that, for all trajectories $\sigma$ from $x \in \abs^{-1}(q)$ to $x' \in \abs^{-1}(q')$, the observation of $p$ along $\sigma$ must belong to $P_Z \cap P_E$.
    Formally, for all $(x, x') \in \abs^{-1}(q) \times \abs^{-1}(q')$,
    \begin{align}\label{eq: PzPe}  
    x' \in \xi(x, \tau)
    \implies
    \forall p \in \AP,
    \chopping{\varsigma}{\tau}(0)(p) \in P_Z(q, q', p) \cap P_E(q, q', p), 
    % \begin{split}
    % &(x \in X \text{ and } x' \in \xi(x, \tau))
    % \\
    % &\implies 
    % \forall p \in \AP,
    % \chopping{\varsigma}{\tau}(0)(p) \in P_Z(q, q', p) \cap P_E(q, q', p),
    % \end{split}
    \end{align}
    where $\chopping{\varsigma}{\tau}$ is the chopped AP-signal of $\varsigma  = P\circ \sigma$, and $\sigma \colon [0,\tau] \to X$ is the finite trajectory  from $x$ to $x'$.
    Then, for all $o \colon \AP \to \OOO$, we require that there exists a transition $(q, o, q') \in \delta$ if
    % \begin{equation}\label{eq:consistency_run}
    % o(p) \in P_Z(q, q', p) \cap P_E(q, q', p), \text{for all } p \in \AP.
    % \end{equation}
    \begin{equation}\label{eq: PzPe2} 
         o(p) \in P_Z(q, q', p) \cap P_E(q, q', p), \text{for all } p \in \AP.
    \end{equation} 
\end{enumerate}
    An intuitive explanation of the two functions is as follows.  
    The function $\rho_Z$ under-approximates the set of atomic propositions that hold and do not hold along the system trajectory at the beginning (at time \textbf{Z}ero).
    We have that $\rho_Z(q, q', p) = +$ if we know $p$ holds at the beginning of any trajectory from $q$ to $q'$ and $\rho_Z(p) = -$ if we know it never holds at the beginning.
    It returns $?$ if the approximation is too imprecise to give an answer or there exist a trajectory where $p$ holds at the beginning and another where it does not.
    The function $\rho_E$ is $\rho_Z$'s counterpart for the \textbf{E}nd of trajectories (at time $\tau$).
    We show example methods to construct these functions in Section~\ref{section:abstraction:discretization}.

% The following theorem ensures that the reduction to symbolic models is sound.
The following theorem states soundness of the reduction to symbolic models.

\begin{theorem}\label{thm:system-abstraction-soundness}
Given a dynamical system $\Sigma$, let $\SSS$ be a symbolic model constructed as above. For a trajectory $\sigma \colon \Rnonneg \to X$ of $\Sigma$, there exists a run $q_0 o_0 q_1 \ldots$ such that, for all $k \in \N$,
$q_k = \gamma(\sigma(k \tau))
\text{ and }
o_k = \chopping{\varsigma}{\tau}(k)$,
% \[
% q_k = \gamma(\sigma(k \tau))
% \text{ and }
% o_k = \chopping{\varsigma}{\tau}(k),
% \]
where $\chopping{\varsigma}{\tau}$ is the chopped AP-signal of $\varsigma = P\circ \sigma \colon \Rnonneg \to 2^{\AP}$. %defined in Section~\ref{section:time-abstraction}. 

% \todo{change}
%     Given a dynamical system $\Sigma$ and a symbolic model $\SSS$ as above, for all symbolic control strategies $\pi_s$ for $\SSS$, there exists a control strategy $\pi$ for $\Sigma$ such that for all trajectories $\sigma_0 \sigma_1 \ldots$ induced by $\pi$, there exists a run $q_0 u_0 o_0 q_1 \ldots$ such that the observation of $p$ on $\sigma_k$ is $o_k$.
\end{theorem}
\begin{proof}
By induction, there exists a run $q_0 o_0 q_1 \ldots$ such that, for all $k \in \N$, we have $q_k = \gamma(\sigma(k \tau))$ by \eqref{eq: approximate simulation}, and $o_k = \chopping{\varsigma}{\tau}(k)$ by \eqref{eq: PzPe} and \eqref{eq: PzPe2}.
\end{proof}

\subsection{System Discretization}
\label{section:abstraction:discretization}

The most common way to construct $\SSS$ from $\Sigma$ is to quantize the system state space $X$ into a discrete finite state set $Q$ using fixed-length grid cells.
The quantization of space is illustrated in Fig.~\ref{fig:quantized}.
Formally, $Q = \setcomp{(k_1 \eta, \ldots, k_n \eta) \in X}{k_i \in \Z}$, and $\abs$ maps each state $X$ to the closest state in $Q$ (with an arbitrary choice for states at equal distance from several points in $Q$).
Using this abstraction process, the following $\rho_Z$ and $\rho_E$ satisfy the requirements in Section~\ref{section:Symbolic Models and Control Strategies}:
% Using the above state abstraction process, one candidate pair of functions $\rho_Z$ and $\rho_E$ that satisfies the requirements in Section~\ref{section:Symbolic Models and Control Strategies} is as follows:
\begin{align*}
    \rho_Z(q, q', p) &=
    \begin{cases}
        + & \text{if for all $x \in \BBB_{\eta/2}(q)$, $P(x)(p) = \top$} \\
        - & \text{if for all $x \in \BBB_{\eta/2}(q)$, $P(x)(p) = \bot$} \\
        ? & \text{otherwise}
    \end{cases}
    %\\
    % \rho_E(q, q', p) &=
    % \begin{cases}
    %     + & \text{if for all $x \in \BBB_{\eta/2}(q')$, $P(x)(p) = \top$} \\
    %     - & \text{if for all $x \in \BBB_{\eta/2}(q')$, $P(x)(p) = \bot$} \\
    %     ? & \text{otherwise,}
    % \end{cases}
\end{align*}
where $\BBB_{\eta/2}(q) = \{x \in \mathbb{R}^n \,\mid\, 
\infinitynorm{q-x} \leq \eta/2\}$, and $\rho_E(q, q', p)$ is defined similarly, replacing
% $q$ by $q'$.
$\BBB_{\eta/2}(q)$ by $\BBB_{\eta/2}(q')$.
The intuition is that $\rho_Z(q, q', p)$ (\emph{resp.} $\rho_E(q, q', p)$) should be $+$ if for all trajectories from $x \in \BBB_{\eta/2}(q)$ to $x' \in \BBB_{\eta/2}(q')$, $p$ holds at the beginning (\emph{resp.} the end)
of the trajectory, 
i.e., exactly when for all $x \in \BBB_{\eta/2}(q)$, $P(x)(p) = \top$. 
% Lemma~\ref{lem: BSZN} and~\eqref{eq: approximate simulation} 
% ensure that 
% $\rho_Z$ and $\rho_E$ satisfy~\eqref{eq: PzPe}, and can be used to construct $\delta$ satisfying~\eqref{eq: PzPe2}.
By Lemma~\ref{lem: BSZN} and~\eqref{eq: approximate simulation}, 
these functions $\rho_Z$ and $\rho_E$ can be used to construct $P_Z$ and $P_E$ satisfying~\eqref{eq: PzPe}, and $\delta$ satisfying~\eqref{eq: PzPe2}. % by over-approximating all possible trajectories.
%by assigning transitions for all possible trajectories.

\section{
\ABZESN{} Automata
}\label{section:ABZESN}
We introduce \emph{\ABZESN{} automata}, where the transitions are labeled by observations of atomic propositions. For a given LTL formula, we construct a generalized \ABZESN{} automaton that approximates all observations of the subformulas. This construction is inspired by Vardi and Wolper's translation of LTL formulas to generalized Büchi automata~\cite{vardi1986automata}, but is specifically adapted to our setting for continuous-time LTL, as we need to consider four observations in $\OOO = \{A, Z, E, N\}$, instead of the two values $\top$ and $\bot$.
%, which impact the construction of the generalized \ABZESN{} automaton.

% Our goal in this section is to adapt Gerth et al.'s algorithm~\cite{} to our setting for continuous-time LTL.
% We need to consider six values for LTL formulas (instead of the two values $\top$ and $\bot$), which impacts the construction of the generalised Büchi automaton.

\subsection{Signal Word and \ABZESN{} Automata}\label{section:signal word}

A \emph{signal word} is a function $w \colon \N \to \OOO^{\AP}$ such that for all $k \in \N$ and $p \in \AP$, $w_k(p) \in \set{A, E}$ iff $w_{k+1}(p) \in \set{A, Z}$ (thus $w_k(p) \in \set{Z, N}$ iff $w_{k+1}(p) \in \set{E, N}$), where $w_k$ is a shorthand for $w(k)$.
Notice that a chopped AP-signal $\chopping{\varsigma}{\tau}$, defined in Section~\ref{section:time-abstraction}, is a signal word.
The intuition is that a signal word is an abstraction of all possible signals mapped to it through signal chopping.

In this section, we assume some given LTL formula $\phi$, and we want to construct an automaton $\BBB_\phi$ that is sound for $\phi$, i.e., it only accepts words that represent signals that satisfy $\phi$.
Formally, we want to build $\BBB_\phi$ such that, if $w$ is in its recognized language, then for all signals $\varsigma$ and durations $\tau$, if $\chopping{\varsigma}{\tau} = w$, then $\varsigma, 0 \vDash \phi$.

Hence, we introduce \ABZESN{} automata, which we use to verify that dynamical systems satisfy continuous-time LTL properties.
They are very similar to classic Büchi automata used for verification of LTL, but one crucial difference is that they work on signal words on the alphabet $\OOO^{\AP}$, rather than words on the alphabet $2^{\AP}$.
Formally, a \emph{nondeterministic \ABZESN{} automaton} (or simply \emph{\ABZESN{} automaton})
is a tuple $\BBB = (B, \delta_b, b_\ini, F)$, where $B$ is a finite set of states, $\delta_b \subseteq B \times \OOO^{\AP} \times B$ is the transition relation, $b_\ini \in B$ is the initial state, and $F \subseteq B$ is the set of accepting states.
A \emph{run} of a signal word $w$ through $\BBB$ is an infinite sequence of states $b_0 b_1 \ldots$ such that $b_0 = b_\ini$ and for all $k \in \N$, $(b_k, w_k, b_{k+1}) \in \delta$.
A run is accepting if it visits $F$ infinitely many times.
The \emph{recognized language} of $\BBB$ is the set of signal words that induce at least one accepting run.
% and whose runs are all accepting.
% Note that, contrary to classic Büchi automata, the nondeterminism of an \ABZESN{} automaton is demonic, rather than angelic: all runs need to be accepting to recognize a signal word, while it is enough for one run to be accepting for a Büchi automaton to recognize a word.

% The automaton in Fig.~\ref{fig:automaton}\todo{need change. Now, Fig. 2 is a generalized automaton} is an example of Büchi automaton $(\setcomp{q_i}{0 \leq i \leq 3}, \delta_b, q_0, \set{q_2, q_3})$ where $\delta_b$ can be derived from the picture.
% For example, $(q_0, p^{N}, q_1) \in \delta_b$, where $p^{N}$ denotes the function that maps $p$ to $N$.

Like the original 
LTL-to-Büchi-automaton
construction~\cite{vardi1986automata}, we start by first building a generalized \ABZESN{} automaton $\AAA_\phi$, then turn it into a (nondeterministic) \ABZESN{} automaton $\BBB_\phi$.
The following construction is the counterpart of generalized Büchi automata.
A \emph{generalized \ABZESN{} automaton} is a tuple $\AAA = (A, \delta_a, a_\ini, \FFF)$, where $\FFF \subseteq \PPP(A)$ %\todo{$2^B$?}
is a set of accepting sets.
All definitions are similar to those of \ABZESN{} automata, except that a run is accepting if it visits all $F \in \FFF$ infinitely often.
Figure~\ref{fig:automaton} shows an example of a generalized Büchi automaton $(\set{q_0, q_1, q_2, q_3}, \delta_b, q_0, \set{\{q_2, q_3\}, \{q_1, q_2, q_3\}})$, where $\delta_b$ can be derived from the picture.
For example, $(q_0, g^{N}, q_1) \in \delta_b$, where $g^{N}$ denotes the observation function that maps $g$ to $N$.

\subsection{Translation to Generalized \ABZESN{} Automaton}

\begin{figure}[t]
    \centering 
    % \resizebox{0.32\textwidth}{0.1\textwidth}
    {
    \begin{tikzpicture}[>=stealth, node distance=2cm, every state/.style={draw, circle, minimum size=0.5cm}]
    \tikzset{every state/.append style={font=\small}, every node/.append style={font=\scriptsize}}
    
    % Nodes
    \node[state, initial, initial text={}] (q0) {$q_0$};
    \node[state, accepting, double distance=1.5pt, right=2.5cm of q0] (q2) {$q_2$};
    \node[state, %double distance=1pt,
    above right=0.5cm and 1.2cm of q0] (q1) {$q_1$};
    \node[state, accepting, double distance=1.5pt, right=of q2] (q3) {$q_3$};

    % Transitions
    \path[->]
        (q0) edge[left] node[above, pos=0.3] {$N$} (q1)
        (q0) edge  node[above] {$Z$} (q2)
        (q0) edge[bend right = 25] node[above] {$A,  E$} (q3)
        (q1) edge[loop left, looseness=6] node[left] {$N$} (q1)
        % (q1) edge[bend right]  node[above] {$S$} (q2)
        (q1) edge[bend left]  node[below] {$E$} (q3)
        (q2) edge node[below] {$N$} (q1)
        % (q2) edge[loop above, looseness=6] node[right] {$S$} (q2)
        (q2) edge[right]  node[above] {$E$} (q3)
        (q3) edge[loop above, looseness=6] node[above] {$A$} (q3)
        (q3) edge[bend right]  node[above] {$Z$} (q2)
        ;
    \end{tikzpicture}
    
    }
    % \todo[inline]{change the figure}
    \caption{
    The generalized \ABZESN{} automaton for verifying the LTL formula $\Globally\Finally \apgreen$ $ (\equiv\bot \Release (\top \Until \apgreen))$ in Section~\ref{section:example}, but minimized by merging equivalent states ($q_3$ is a merged state).
    % where $\FFF = \{F_{\Finally \apgreen}=\{q_2, q_3\}, F_{\Globally\Finally \apgreen} = \{q_1, q_2, q_3\}\}$. In this case, 
    The subformulas are $\apgreen$, $\Finally \apgreen$, and $\Globally\Finally \apgreen$, and the accepting state sets are $\FFF = \{F_{\Finally \apgreen}=\{q_2, q_3\}, F_{\Globally\Finally \apgreen} = \{q_1, q_2, q_3\}\}$.
    % The valuation at
    % $q_1$ for the three subformulas are, $N$, $A$, and $A$, in order.  
    % The valuation at
    % $q_2$ is $(Z, A, A)$, and those at $q_3$ are $(A, A, A)$ and $(E, A, A)$. 
    Apart from $q_0$, the only non-accepting state (for $\Finally \apgreen$) is $q_1$. This reflects the fact that the system violates the specification only if its trajectory never visits a region satisfying $\apgreen$ after some point, i.e., the observation along the transitions is always $N$ ($\apgreen$ is never satisfied) after that point.
    % ~
    % The $\ABZESN$ automaton for verifying the LTL formula $\Globally\Finally \apgreen$ $ (\equiv\bot \Release (\top \Until \apgreen))$ in Section~\ref{section:example}, but minimized by merging equivalent states ($q_3$ is a merged state). 
    % In this case, the subformulas are $\apgreen$, $\Finally \apgreen$, and $\Globally\Finally \apgreen$.
    % The valuation at
    % $q_1$ for the three subformulas are, $N$, $A$, and $A$, in order.  
    % The valuation at
    % $q_2$ is $(Z, A, A)$, and that at $q_3$ are $(A, A, A)$ and $(E, A, A)$. 
    % Notice that, apart from $q_0$, the only non-accepting state is $q_1$. This reflects the fact that the system violates the specification only if its trajectory never visits a region satisfying $\apgreen$ after some point, i.e., the observation along the transitions is always $N$ ($\apgreen$ is never satisfied.) after that point.
    }
    \label{fig:automaton}
\end{figure}

Given an LTL formula $\phi$,
the construction of the corresponding generalized \ABZESN{} automaton $\AAA_\phi$ relies on the set $\cs(\phi)$ of \emph{consistent subformula valuations} $\nu \colon \sub(\phi) \to \OOO$ of $\phi$, where $\sub(\phi)$ is the set of subformulas of $\phi$. We first present the construction of the automaton and define $\cs(\phi)$ later. 
$\AAA_\phi = (A_\phi = \cs(\phi) \cup \set{q_0}, \delta_\phi, q_0, \FFF_\phi)$ is constructed as follows.

\begin{itemize}
% \item $G_\phi = \cs(\phi) \cup \set{q_0}$ is the set of states.
%, where $\cs(\phi)$ is the set of consistent subformula valuations of $\phi$ (explained below).
\item For all states $\nu \in \cs(\phi)$, $(\nu, o, \nu') \in \delta_\phi$ iff  
(A1) for all $p \in \AP$, $\nu'(p) = o(p)$, and
(A2) for all subformulas $\psi \in \sub(\phi)$, $\nu(\psi) \in \set{A, E}$ iff $\nu'(\psi) \in \set{A, Z}$.  
% \begin{enumerate}
%     \item[(A1)] for all $p \in \AP$, $\nu'(p) = o(p)$, and
%     \item[(A2)] for all sub-formulas $\psi \in \sub(\phi)$, $\nu(\psi) \in \set{A, E}$ iff $\nu'(\psi) \in \set{A, Z}$.
% \end{enumerate} 
\item 
There exists  $(q_0, o, \nu') \in \delta_\phi$ from the initial state $q_0$ iff 
(B1) for all $p \in \AP$, $\nu'(p) = o(p)$, and 
(B2) $\nu'(\phi) \in \set{A, Z}$.
% \begin{enumerate}
%     \item[(B1)] for all $p \in \AP$, $\nu'(p) = o(p)$, and 
%     \item[(B2)] $\nu'(\phi) \in \set{A, Z}$.
% \end{enumerate} 
% \item 
% There exists  $(q_0, o, \bot) \in \delta_\phi$ from the initial state $q_0$ to $\bot$ iff $\nu'(\phi) \in \set{E, S, N}$
% \item 
% There exists $(\bot, o, \bot) \in \delta_\phi$ for all $o$.
\item 
% If we denote $F_{\psi_1 \Until \psi_2} = \setcomp{\nu \in \cs(\phi)}{\nu(\psi_2) \neq N \text{ or } \nu(\psi_1 \Until \psi_2) \neq A}$ and
% $F_{\psi_1 \Release \psi_2} = \setcomp{\nu \in \cs(\phi)}{\nu(\psi_2) \neq A \text{ or } \nu(\psi_1 \Release \psi_2) \neq N}$, then the set of accepting state-sets is
% \begin{equation*}\label{eq:accepting_cond}
%     \FFF_\phi = \setcomp{F_{\psi_1 \Until \psi_2}}{\psi_1 \Until \psi_2 \in \sub(\phi)} \cup \setcomp{F_{\psi_1 \Release \psi_2}}{\psi_1 \Release \psi_2 \in \sub(\phi)},
% \end{equation*}
%The set of sets of accepting states is
% \begin{equation}\label{eq:accepting_cond}
$\FFF_\phi = \setcomp{F_{\psi_1 \Until \psi_2}}{\psi_1 \Until \psi_2 \in \sub(\phi)} \cup \setcomp{F_{\psi_1 \Release \psi_2}}{\psi_1 \Release \psi_2 \in \sub(\phi)}$ is the set of accepting states,
% \end{equation} 
where
$F_{\psi_1 \Until \psi_2} = \setcomp{\nu \in \cs(\phi)}{\nu(\psi_2) \neq N \text{ or } \nu(\psi_1 \Until \psi_2) \neq A}$ and
$F_{\psi_1 \Release \psi_2} = \setcomp{\nu \in \cs(\phi)}{\nu(\psi_2) \neq A \text{ or } \nu(\psi_1 \Release \psi_2) \neq N}$.
\end{itemize}

\begin{figure}[t]
\small
    \centering
    \[
    \begin{array}{|c|c|c|c|c|c|}
      \hline
      \psi_1 & \psi_2 & \psi_1 \land \psi_2 & \psi_1 \lor \psi_2 & \psi_1 \Until
        \psi_2 & \psi_1 \Release \psi_2 \\
      \hline 
      \multirow{4}{*}{\color{blue}\underline{\textbf{A}}}
      & A & A & A & A & A \\
      \cline{2-6}
      & \underline{Z} & \underline{Z} & A & \underline{AZ} & Z \\
      \cline{2-6}
      & E & E & A & A & E \\
      \cline{2-6}
      & \color{blue}\textbf{N} & N & A & \color{blue}\textbf{AN} & N \\
      \hline
      \multirow{4}{*}{$Z$}
      & A & Z & A & A & AZ \\
      \cline{2-6}
      & Z & Z & Z & Z & Z \\
      \cline{2-6}
      & E & N & A & A & EN \\
      \cline{2-6}
      & N & N & Z & N & N \\
      \hline
    \end{array}
    \hspace{1em}
    \begin{array}{|c|c|c|c|c|c|}
      \hline
      \psi_1 & \psi_2 & \psi_1 \land \psi_2 & \psi_1 \lor \psi_2 & \psi_1 \Until
        \psi_2 & \psi_1 \Release \psi_2 \\
      \hline 
      \multirow{4}{*}{$E$}
      & A & E & A & A & A \\
      \cline{2-6}
      & Z & N & A & AZ & N \\
      \cline{2-6}
      & E & E & E & E & E \\
      \cline{2-6}
      & N & N & E & EN & N \\
      \hline
      \multirow{4}{*}{\color{blue}\textbf{N}}
      &\color{blue}\textbf{A} & N & A & A & \color{blue}\textbf{AN} \\
      \cline{2-6}
      & Z & N & Z & Z & N \\
      \cline{2-6}
      & E & N & E & E & EN \\
      \cline{2-6}
      & N & N & N & N & N \\
      \hline
    \end{array}
    % \begin{array}{|c|c|c|c|c|c|}
    %   \hline
    %   \psi_1 & \psi_2 & \psi_1 \land \psi_2 & \psi_1 \lor \psi_2 & \psi_1 \Until
    %     \psi_2 & \psi_1 \Release \psi_2 \\
    %   \hline 
    %   \multirow{4}{*}{\color{blue}\underline{\textbf{A}}}
    %   & A & A & A & A & A \\
    %   \cline{2-6}
    %   & \color{blue}\textbf{Z} & \color{blue}\textbf{Z} & A & \color{blue}\textbf{AZ} & Z \\
    %   \cline{2-6}
    %   & E & E & A & A & E \\
    %   \cline{2-6}
    %   & \underline{N} & N & A & \underline{AN} & N \\
    %   \hline
    %   \multirow{3}{*}{$Z$}
    %   & A & Z & A & A & AZ \\
    %   \cline{2-6}
    %   & Z & Z & Z & Z & Z \\
    %   \cline{2-6}
    %   & N & N & Z & N & N \\
    %   \hline
    %   \multirow{3}{*}{$E$}
    %   & A & E & A & A & A \\
    %   \cline{2-6}
    %   & E & E & E & E & E \\
    %   \cline{2-6}
    %   & N & N & E & EN & N \\
    %   \hline
    %   \multirow{4}{*}{$N$}
    %   &A & N & A & A & AN \\
    %   \cline{2-6}
    %   & Z & N & Z & Z & N \\
    %   \cline{2-6}
    %   & E & N & E & E & EN \\
    %   \cline{2-6}
    %   & N & N & N & N & N \\
    %   \hline
    % \end{array}
  \]
    \caption{The consistency rules for generalized \ABZESN{} automata. 
    }
    \label{fig:consistency}
\end{figure}

Next, we define $\cs(\phi)$.  
A \emph{subformula valuation} is a function $\nu \colon \sub(\phi) \to \OOO$.
%, where $\sub(\phi)$ is the set of all subformulas of $\phi$.
For example,
in Fig.~\ref{fig:automaton}, the valuations of all states map $\Finally g (\equiv \top \Until g)$ and $\Globally \Finally g (\equiv \bot \Release \Finally g)$ to $A$, state $q_1$ maps $g$ to $N$, $q_2$ maps it to $Z$, and $q_3$ (which is a merge of two states) maps it either to $A$ or $E$.
% \ssc{Then, we leverage the concept of \emph{consistent} properties for $AP$-observations in} Eq.~\eqref{eq:consistency_observation} to apply to the valuation on subformulas of the given LTL formula.
The valuation $\nu$ is \emph{consistent} if $\nu$ follows the rules given by Fig.~\ref{fig:consistency}.
% \begin{itemize}
%     \item \ssc{for all $p \in \AP$, $\nu(p)$ and $\nu(\neg p)$} satisfy \eqref{eq:consistency_observation}, and
%     \item $\nu$ follows the rules given by Fig.~\ref{fig:consistency}.
% \end{itemize}
Then, $\cs(\phi)$ is the set of consistent subformula valuations of $\phi$. The size of the automaton is exponential in the number of subformulas.

The way to read the table in Fig.~\ref{fig:consistency} is as follows: given observations for subformulas $\psi_1$ and $\psi_2$, the consistent observations of $\psi_1 \land \psi_2$, $\psi_1 \lor \psi_2$, $\psi_1 \Until \psi_2$, and $\psi_1 \Release \psi_2$ are given in the table.
The intuition is that a valuation $\nu$ represents the current state of all subformulas.
For example, let us consider the second row of the table (the one with underlined text). 
It states in particular that, if $\nu(\psi_1) = A$ and $\nu(\psi_2) = Z$, then $\nu(\psi_1 \land \psi_2) = Z$ and $\nu(\psi_1 \Until \psi_2) \in \set{A, Z}$.
Indeed, if $\psi_1$ holds on a whole time interval and $\psi_2$ holds at its beginning but not at the end, then $\psi_1 \land \psi_2$ also holds at the beginning and not at the end, and $\psi_1 \Until \psi_2$ either holds only at the beginning (if $\psi_2$ never holds again, or $\psi_1$ stops holding before $\psi_2$ holds again) or on the whole interval (if $\psi_1$ keeps holding until $\psi_2$ holds again).

A noteworthy case is in the fourth row of the table (in bold blue text):  $\nu(\psi_1) = A$, $\nu(\psi_2) = N$, and $\nu(\psi_1 \Until \psi_2) \in \set{A, N}$.
If $\psi_1$ holds at all times and $\psi_2$ holds at none of the interval, then $\psi_1 \Until \psi_2$ either holds for the whole interval (if $\psi_1$ keeps holding until $\psi_2$ holds) or does not hold on the interval (if $\psi_1$ stops holding before $\psi_2$ holds).
Notice that a run $\nu_1 \nu_2 \ldots$ that assigns $\nu_k(\psi_1 \Until \psi_2) = A$ and $\nu_k(\psi_2) = N$ at all time steps $k$ is not an accepting run thanks to $F_{\psi_1 \Until \psi_2}$.
The reasoning is the same for $\Release$ when $\nu(\psi_1) = N$ and $\nu(\psi_2) = A$.

% \ce{Note in particular that the table has no rows $\psi_1 = Z$, $\psi_2 = E$ or $\psi_1 = E$, $\psi_2 = Z$.
% That is because it can never be the case that $\psi_1$ holds at the beginning of an interval and $\psi_2$ at the end by Assumption~\ref{asm: restriction BSZE}.}
We remark that the two rows $(\nu(\psi_1) = Z, \nu(\psi_2) = E)$ and $(\nu(\psi_1) = E, \nu(\psi_2) = Z)$ use the fact that only one atomic proposition may change during a time interval of length $\tau$ by Assumption~\ref{asm: restriction BSZE}, so if the satisfaction of two subformulas change during that interval, they must change exactly at the same point in time.

Formally, the connector $\land$ comes equipped with a function $c_\land \colon \OOO \times \OOO \to 2^{\OOO}$ described by the third column of Fig.~\ref{fig:consistency} (and similarly for connectors $\lor$, $\Until$, and $\Release$).
A subformula valuation $\nu \colon \sub(\phi) \to \OOO$ is \emph{consistent} if for all $\psi_1, \psi_2 \in \sub(\phi)$ such that $\psi_1 \odot \psi_2 \in \sub(\phi)$, $\nu(\psi_1 \odot \psi_2) \in c_\odot(\nu(\psi_1), \nu(\psi_2))$ for all connectors $\odot \in \set{\land, \lor, \Until, \Release}$.

The following lemma states that the table in Fig.~\ref{fig:consistency} is sound and complete. Note that the AP-signals $\varsigma  $ referred to in this lemma are general, and not necessarily those generated by $\Sigma$.% maybe remove later
\begin{lemma}
    \label{lem:consistency}
    For all connectors $\odot \in \set{\land, \lor, \Until, \Release}$, formulas $\psi = \psi_1 \odot \psi_2$, and signals $\varsigma \colon \Rnonneg \to 2^{\AP}$ and $\tau \in \Rpos$ that satisfy Assumption~\ref{asm: restriction BSZE},
    $\actualchopping{\varsigma}{\tau}(k)(\psi) \in c_{\odot}(\actualchopping{\varsigma}{\tau}(k)(\psi_1), \actualchopping{\varsigma}{\tau}(k)(\psi_2))$ for all $k \in \N$.
    
    Moreover, for all $\psi = \psi_1 \odot \psi_2$, $o_1, o_2 \in \OOO$ and $o \in c_{\odot}(o_1, o_2)$, there exists a signal $\varsigma   \colon \Rnonneg \to 2^{\AP}$ and $\tau \in \Rpos$  that satisfy Assumption~\ref{asm: restriction BSZE} and such that, for all $k \in \N$, $\actualchopping{\varsigma}{\tau}(k)(\psi_i) = o_i$ for $i \in \set{1, 2}$ and $\actualchopping{\varsigma}{\tau}(k)(\psi) = o$.
\end{lemma}

The proof can be found in~\switchversion{\cite[Appendix~A.2]{pruekprasert2025}}{Appendix~\ref{proof:lem:consistency}}.
%~\cite[Appendix~\switchversion{A.2}{\ref{proof:lem:consistency}}]{pruekprasert2025}.
The following theorem proves that the generalized \ABZESN{} automaton construction of $\AAA_\phi$ is sound.

\begin{theorem}
    \label{thm:abzesn-soundness}
    For all continuous-time LTL formulas $\phi$, AP-signals $\varsigma \colon \Rnonneg \to 2^{\AP}$, and durations $\tau$, if $\chopping{\varsigma}{\tau}$ is in the recognized language of $\AAA_\phi$, then $\varsigma, 0 \vDash \phi$.
\end{theorem}

We relegate the full proof to
\switchversion{the appendix~\cite[Appendix~A.3]{pruekprasert2025}}{Appendix~\ref{proof:lem:ap-word-extension-accepting-run}}
and only state a few crucial lemmas. %to facilitate understanding of the construction.
Lemma~\ref{lem:ap-word-extension-accepting-run} gives a fundamental property of $\AAA_\phi$: given a word, there exists exactly one non-initial state and one accepting run along that word from that state.
Its proof
%in~\switchversion{\cite[Appendix~A.3]{pruekprasert2025}}{Appendix~\ref{proof:lem:ap-word-extension-accepting-run}}
heavily relies on Lemma~\ref{lem:consistency} in order to show by induction on subformulas $\psi$ that there is a unique possible value for $\nu_k(\psi)$.

\begin{lemma}
    \label{lem:ap-word-extension-accepting-run}
    For all words $w \colon \N \to \OOO^{\AP}$ such that
    \begin{align}
        \forall k \in \N.\ \forall p \in \AP.\ w_k(p) \in \set{A, E} \iff w_{k+1}(p) \in \set{A, Z},
        \label{eq:ap-word-extension-accepting-run:ap-transition} \\
        \forall k \in \N.\  \forall p, p' \in \AP.\ w_k(p), w_k(p') \in \set{Z, E} \Rightarrow p = p',
        \label{eq:ap-word-extension-accepting-run:no-ZE}
    \end{align}
    there exists a unique accepting run $\nu_0 \nu_1 \ldots$ such that for all $k \in \N$ and $p \in \AP$, $\nu_k(p) = w_k(p)$.
\end{lemma}

The following corollary demonstrates that accepting runs are, in fact, exactly valuations of signal choppings.
Its proof is relegated to~\switchversion{\cite[Appendix~A.4]{pruekprasert2025}}{Appendix~\ref{proof:cor:accepting-word-signal}},
% ~\cite[Appendix~\switchversion{A.4}{\ref{proof:cor:accepting-word-signal}}]{pruekprasert2025}, 
but it crucially uses Assumption~\ref{asm: restriction BSZE} to show that signal choppings have the same shape as the accepting runs exhibited in Lemma~\ref{lem:ap-word-extension-accepting-run}.
% but it crucially uses Assumption~\ref{asm: restriction BSZE} to show that signal choppings have the shape described in the proof of Lemma~\ref{lem:ap-word-extension-accepting-run}.

\begin{corollary}    
    \label{cor:accepting-word-signal}
    Given $\varsigma \colon \Rnonneg \to 2^{\AP}$ and $\tau \in \Rpos$, a run $\nu_0 \nu_1 \ldots$ such that $\nu_k(p) = \actualchopping{\varsigma}{\tau}(k)(p)$, for all $k \in \N$ and $p \in \AP$, is an accepting run iff $\actualchopping{\varsigma}{\tau}(k)(\psi) = \nu_k(\psi)$, for all $k \in \N$ and $\psi \in \sub(\phi)$.
\end{corollary}

Theorem~\ref{thm:abzesn-soundness} follows directly from Corollary~\ref{cor:accepting-word-signal}, using the shape of transitions from the initial state, as we only have transitions from $q_0$ to $\nu$ with $\nu(\phi) \in \{A, Z\}$.

As a side result, Corollary~\ref{thm:abzesn-soundness} and the proof of Lemma~\ref{lem:ap-word-extension-accepting-run} can be used to show that, given a word of observations $w_k$ for atomic propositions, there exists a unique word of observations $\nu_k$ for all formulas compatible with $w_k$ (regardless of the \ABZESN{} automaton considered).
This can in turn be used to define whether a symbolic model satisfies a formula and prove that the construction is sound and complete \emph{for symbolic models}.
However, due to the loss of information during discretization, the construction is not complete for dynamical systems.

Theorem~\ref{thm:abzesn-soundness} ensures that the construction of $\AAA_\phi$ is sound.
This construction is inspired by  the 
LTL-to-Büchi-automaton construction by Vardi and Wolper~\cite{vardi1986automata}, which we briefly discussed in Section~\ref{section:System Verification for LTL specifications}.
However, there is a fundamental difference in that there are no explicit constraints on transitions.
Indeed, in the original construction, where states are consistent valuations $\nu \colon \sub(\phi) \to 2$, there can be a transition from $\nu$ to $\nu'$ only if they ``agree'' on the value of all formulas $\psi_1 \Until \psi_2$ and $\psi_1 \Release \psi_2$.
This uses the fact that, in discrete time, $\psi_1 \Until \psi_2 \iff \psi_2 \lor (\psi_1 \land \Next (\psi_1 \Until \psi_2))$, so for example if $\psi_1 \Until \psi_2$ holds in $\nu$, then either $\psi_2$ should hold in $\nu$, or $\psi_1$ should hold in $\nu$ and $\psi_1 \Until \psi_2$ in $\nu'$.
Similarly, for Release, using the fact that $\psi_1 \Release \psi_2 \iff \psi_2 \land (\psi_1 \lor \Next (\psi_1 \Release \psi_2))$.

In our translation, this constraint comes from the fact that the consistency rules in Fig.~\ref{fig:consistency} also contain constraints on Until and Release, while the original translation only has constraints on conjunction and disjunction.
This, coupled with a generalization of~\eqref{eq:ap-word-extension-accepting-run:ap-transition} to all subformulas, gives constraints between valuations in $\nu$ and $\nu'$.
However, we need to add another constraint to only retain good behaviors.
Indeed, while the original construction only has accepting sets for Until subformulas, here we also need to add accepting sets for the Release subformulas to make up for the constraints in the original construction.

Following the translation from generalized Büchi automata to nondeterministic Büchi automata~\cite{gastin2001fast}, we can translate 
the generalized \ABZESN{} automata $\AAA_\phi$ to their \emph{nondeterministic \ABZESN{} automata} counterpart $\BBB_\phi$, from which it is easier to build a game-based verification algorithm. 
% generalized \ABZESN{} automata to their \emph{nondeterministic \ABZESN{} automata} counterparts. 
In our implementation for the example in Section~\ref{section:example}, before applying this translation, we prune the generalized \ABZESN{} automaton $\AAA_\phi$ by removing states that cannot lead to any accepting run (i.e., that fail to reach at least one state in each $F \in \FFF(\AAA_\phi)$), and merge states that are equivalent with respect to acceptance conditions and outgoing transitions. The generalized \ABZESN{} automaton shown in Fig.~\ref{fig:automaton} reflects the outcome of this pruning and minimization.

\section{System Verification}
\label{section:verification}

We can build a symbolic model $\SSS = (Q, \delta_s, q_\ini)$ that over-approximates the behaviors of the dynamical system $\Sigma$ in the sense of Theorem~\ref{thm:system-abstraction-soundness} and a nondeterministic \ABZESN{} automaton $\BBB_\phi = (B, \delta_b, b_\ini, F)$ whose language is that of $\phi$ in the sense of Theorem~\ref{thm:abzesn-soundness}.
Both $\SSS$ and $\BBB_\phi$ are nondeterministic.
The nondeterminism in $\SSS$ is demonic and comes from that of $\Sigma$: if $(q, o_1, q_1), (q, o_2, q_2) \in \delta_s$, then we cannot choose whether the system goes to $q_1$ by reading $o_1$ and to $q_2$ by reading $o_2$.
The nondeterminism in $\BBB_\phi$ is angelic: a word is recognized if there exists an accepting run.
We mix these two forms of nondeterminism using a Büchi game~\cite{chatterjee2014efficient}, %which is 
a particular type of parity game~\cite{zielonka1998infinite} (with parities 1 and 2).
% In order to mix these two forms of nondeterminism, we build a Büchi game~\cite{chatterjee2014efficient}, %which is 
% a particular type of parity game~\cite{zielonka1998infinite} (with parities 1 and 2).

A \emph{Büchi game} is a tuple $\GGG = (G, G_0, \delta_g, \GameF, g_\ini)$, where $G$ is a set of \emph{vertices}, $G_0 \subseteq G$ is the set of \emph{Player} vertices, $G_1 = G \setminus G_0$ is that of \emph{Opponent} vertices, $\delta_g \subseteq G \times G$ is a set of \emph{edges}, and $\GameF \subseteq G$ is the set of \emph{Büchi vertices}.
A \emph{play} is a sequence $g_0 g_1 \ldots$ of states such that $(g_i, g_{i+1}) \in \delta_g$ for all $i \in \N$.
A play is \emph{accepting} if it reaches $\GameF$ infinitely often.
%or ends in $g_i \in G_1$ such that there are no $(g_i,g') \in \delta_g$.
%A \emph{Player strategy} is a function $\strat \colon G_0 \to G$ such that for all $g \in G_0$, $(g, \strat(g)) \in \delta_g$, and similarly for \emph{Opponent strategy}.
A \emph{Player strategy} is a function $\strat_0 \colon G_0 \to G$ such that for all $g \in G_0$, $(g, \strat_0(g)) \in \delta_g$. An \emph{Opponent strategy} $\strat_1 \colon G_1 \to G$ is defined similarly.
The play \emph{induced} by a Player strategy $\strat_0$ and an Opponent strategy $\strat_1$ from a state $g_0$ is the sequence $g_0 g_1 \ldots$ such that for all $i \in \N$, if $g_i \in G_k$, then $\strat_k(g_i) = g_{i+1}$.
A Player strategy $\strat_0$ is \emph{winning} from $g$ if for all Opponent strategies $\strat_1$, the play induced by $\strat_0$ and $\strat_1$ from $g$ is accepting.
A state $g$ is winning if there exists a winning Player strategy from $g$, and $\GGG$ is winning if $g_\ini$ is winning.

Given a nondeterministic \ABZESN{} automaton $\BBB$ and a symbolic model $\SSS$, we build 
$\GGG_{\SSS \times \BBB} = (G, G_0, \delta_g, \GameF = \setcomp{(q, b)}{b \in F}, g_\ini = (q_\ini, b_\ini))$ as follows:
\begin{itemize}
    \item
    $G = \setcomp{(q, b)}{q \in Q, b \in B} \cup \setcomp{(q, o, b)}{q \in Q, o \in \OOO, b \in B}$,
    \item
    $G_0 = \setcomp{(q, o, b)}{q \in Q, o \in \OOO, b \in B}$,
    \item
    $((q, b), (q', o, b)) \in \delta_g$ if and only if $(q, o, q') \in \delta_s$, and $((q, o, b), (q, b')) \in \delta_g$ if and only if $(b, o, b') \in \delta_b$.
    % \item
    % $\GameF = \setcomp{(q, b)}{b \in F}$.
    % \item
    % $g_\ini = (q_\ini, b_\ini)$.
\end{itemize}

\begin{theorem}
    \label{thm:game-soundness}
    Given a symbolic model $\SSS$ an \ABZESN{} automaton $\BBB$, $\GGG_{\SSS \times \BBB}$ is winning iff all plays $q_\ini \xto{o_0} q_1 \xto{o_1} \ldots$ of $\SSS$ are such that $o_0 o_1 \ldots$ is in the recognized language of $\BBB$.
\end{theorem}

\begin{proof}
    It is well-known that positional strategies are optimal~\cite{zielonka1998infinite}.
    In particular, if a positional Player strategy $\strat_0$ wins against all positional Opponent strategies, then it wins against all (general) Opponent strategies $\strat_1 \colon G^* G_1 \to G$ that map each play to a next state.
    Assuming that $\strat_0$ is a winning strategy, given a run $q_\ini \xto{o_0} q_1 \xto{o_1} \ldots$ of $\SSS$, we define
    \[
        \strat_1((q_\ini, b_\ini), (q_1, o_0, b_\ini), \ldots, (q_n, b_n)) = (q_{n+1}, o_n, b_n).
    \]
    Because $\strat_0$ wins against $\strat_1$, the induced play visits $\GameF = \setcomp{(q, b)}{b \in F}$ infinitely often, so $b_\ini \xto{o_0} b_1 \xto{o_1} \ldots$ is an accepting run of $\BBB$, and therefore $o_0 o_1 \ldots$ is in the recognized language of $\BBB$.
\end{proof}

Putting Theorems~\ref{thm:system-abstraction-soundness},~\ref{thm:abzesn-soundness}, and~\ref{thm:game-soundness} together, we get the following corollary.

\begin{corollary}
    \label{cor:verification}
    Given a dynamical system $\Sigma$, a symbolic model $\SSS$ that soundly represents $\Sigma$ (constructed as in Section~\ref{section:Symbolic Models and Control Strategies}), and a continuous-time LTL formula $\phi$, if $\GGG_{\SSS \times \BBB_\phi}$ is winning, then for all trajectories $\sigma$ of $\Sigma$ from $x_\ini$, $\sigma, 0 \vDash \phi$.
\end{corollary}

This gives a sound algorithm for Problem~\ref{problem:continuousVerification}.
If the game is winning, then all trajectories $\sigma$ of $\Sigma$ from $x_\ini$ are such that $\sigma, 0 \vDash \phi$, but otherwise we cannot conclude that there exists a trajectory such that $\sigma, 0 \nvDash \phi$.
The completeness of this result cannot be guaranteed, as the exact values of atomic propositions along trajectories are lost during the discretization of the system’s state space.
To mitigate this theoretical limitation, we demonstrate the practical feasibility of our approach by verifying several specifications in Section~\ref{section:example}.

\section{Illustrative Example}
\label{section:example}

To illustrate a potential application of the proposed structure, we present an illustrative application example of verifying a surveillance drone system.
As shown in Fig.~\ref{fig_ex_system}, the drone flies in an area of $33\times 33 \text{m}^2$  and can move in eight directions: the four cardinal directions and the four diagonal directions, following linear dynamics at a speed of $4 \text{m/s}$.
% However, the drone's speed ranges between $3.9$ and $4.1 \text{m/s}$, and its direction can deviate by up to $0.08$ radians, introducing nondeterminism into the system.
The system's nondeterminism comes from environmental disturbances impacting the drone's speed and direction. Its speed may vary by up to $0.1 \text{m/s}$ and its angle by up to $0.08$ radians.

\begin{figure}[t]
\centerline{
\includegraphics[width=0.65\textwidth]{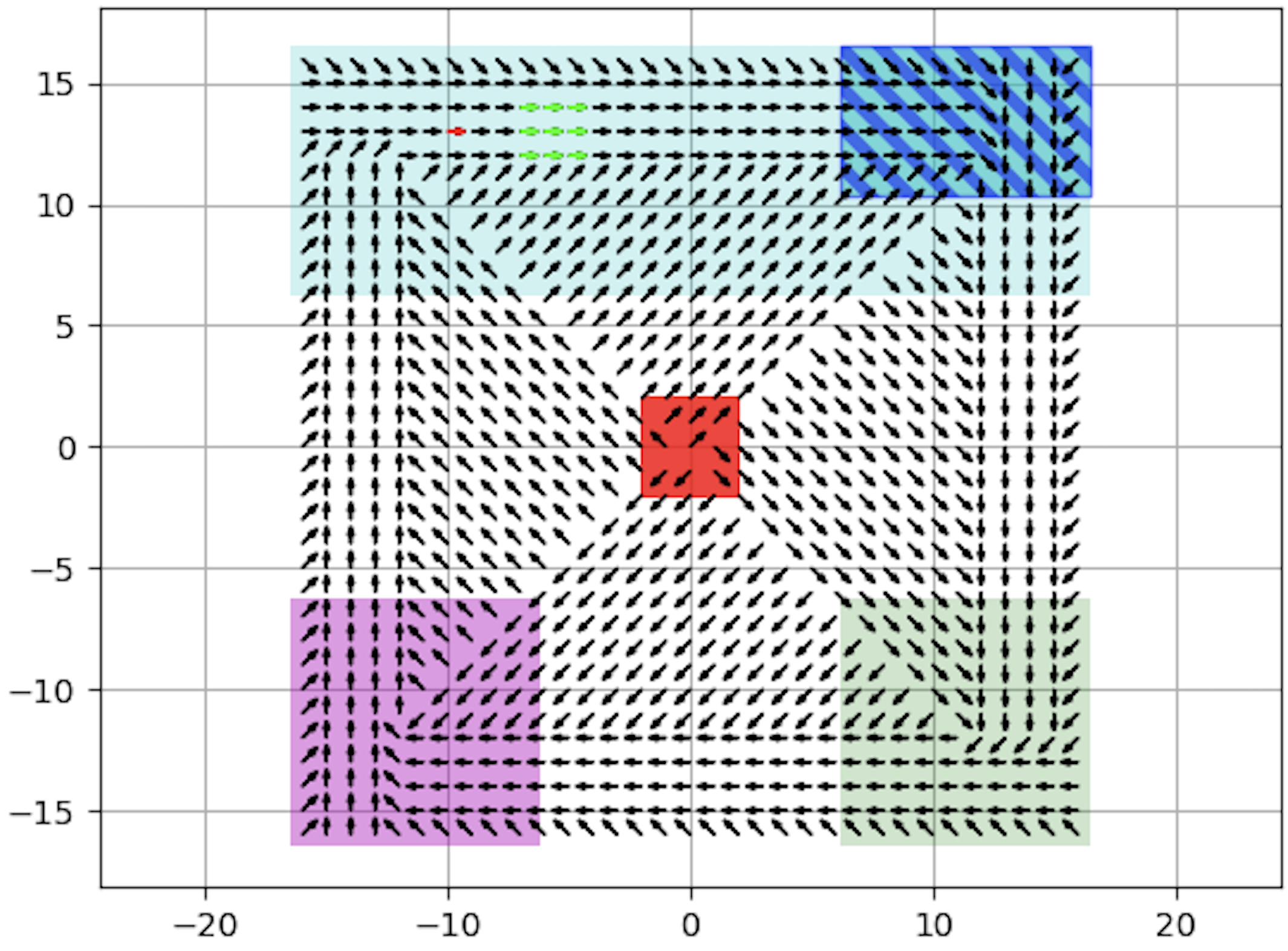}}
\caption{A surveillance drone system flying from the position $(-10, 13)$ (base of the red arrow). 
The atomic propositions are assigned to $xy$-coordinates as follows: 
$\apcyan$ (cyan) for $y \geq 6.21$, 
$\apblue$ (blue) for $(x,y) \geq (6.21, 10.32)$,
%$x \geq 6.21$ and $y \geq 10.32$, 
$\appink$ (pink) for $(x,y) \leq (6.21, 6.21)$,
%$x \leq 6.21$ and $y \leq 6.21$, 
$\apgreen$ (green) for $x \geq 6.21$ and $y \leq 6.21$,
and
$\apred$ (not red) for $|x| > 2.1$ and $|y| > 2.1$. 
The states reached in a single transition from the initial state are at the bases of the green arrows.}
\label{fig_ex_system}
\end{figure}

We construct a symbolic model $\SSS$ using $\eta = 1$m and $\tau=1$s (which satisfies Lemma~\ref{lem:bounded speed}), and fix the moving direction for each symbolic state as depicted by the arrows in Fig.~\ref{fig_ex_system}. 
% \ssc{Notice that using $\tau = 1\text{s}$ satisfies the condition in} Lemma~\ref{lem:bounded speed} \ssc{with $\Delta(t)=4.1t$, as $4.1\,\text{m/s}$ is the maximum possible speed.}
We implement Zielonka's algorithm~\cite{zielonka1998infinite} to solve the game described in Section~\ref{section:verification}.
% Hence, we can consider the controlled symbolic model $\SSS_{\pi_s}$ as an automaton (see Section~\ref{section:Symbolic Models and Control Strategies}).
We implement the algorithm in Python~3.11.2 and run it
on a MacBook Pro (Apple M2 Max, 64GB). 
We verify the formulas in Fig.~\ref{fig:results}.
% We verify several LTL formulas as shown in Fig.~\ref{fig:results}.
% Our algorithm effectively verifies several LTL formulas as shown in Fig.~\ref{fig:results}. 
% Remark that, for some specifications, subformulas could be verified separately to speed up the computation, but we verify the entire formula to demonstrate the feasibility of our approach. 
% Note that all conjunctions and disjunctions of the formulas (e.g., $\Finally \appink \wedge \Globally \neg \mathit{red} \wedge \Globally\Finally \apgreen$) are also verified.

% For complex specifications, the bottleneck\todo{now, games take longer} of the program lies in the construction of the $\ABZESN{}$ automaton.
% Our implementation considers all possible valuations of all subformulas, whose number is exponential in the number of subformulas in the LTL formula. 
% For example, the formula %$\Globally \apred \land  \Finally \appink \land \Finally \apcyan$ 
% $\Globally \apred \land \Finally(\apgreen\land \Finally\appink)$
% has ten subformulas (including $\top$ and $\bot$). Therefore, we consider all $4^{10}$ valuations and check if they are consistent before pruning deadlocked states and minimizing the automaton. 
% This construction serves as a simple baseline,
% % This is a naive construction, 
% and we leave it as future work to optimize the implementation, for example by only considering the states that can be reached from the initial state.

For complex specifications, the main bottleneck lies in the construction of the \ABZESN{} automaton, whose size is exponential in the number of subformulas.
% For complex specifications, the main bottleneck lies in the construction of the \ABZESN{} automaton.
% This step is  
% computationally expensive due to the enumeration of all possible valuations of subformulas, which grows exponentially with the size of the formula. 
For example, the formula  
$\Globally \apred \land \Finally(\apgreen\land \Finally\appink)$
has ten subformulas. %(including $\top$ and $\bot$).
% Therefore, we consider all $4^{10}$ valuations and check if they are consistent before pruning deadlocked states and minimizing the automaton. 
Here, we check the consistency of all $4^{10}$ valuations before pruning deadlocked states and minimizing the automaton. 
This implementation serves as a simple baseline to demonstrate the feasibility of the method, and we leave it as future work to optimize it, for example, by considering only valuations at reachable states. 
%using more recent game-solving algorithms such as that in~\cite{chatterjee2014efficient}, 
% and applying heuristic methods such as that in~\cite{pruekprasert2021fast}.
% This implementation serves as a simple baseline to demonstrate the feasibility of the method, and we leave it as future work to optimize it, for example by only considering the valuations at reachable states, and applying heuristic methods, such as one proposed in~\cite{pruekprasert2021fast}.

% Despite the automata containing up to 4,144 states, the verification process is completed in less than one minute. 

\begin{figure}[t!]
\centering

\begin{tabular}{|>{\centering\arraybackslash}m{2.5cm}|c|c|c|c|c|c|c|}
\hline 
\textbf{Specification}
& \multicolumn{2}{c|}{{{$\BBB_\phi$}}} 
& \multicolumn{2}{c|}{{Game construction}}
& \multicolumn{1}{c|}{{Game solving}
} 
& \multicolumn{1}{c|}{\multirow{2}{*}{\shortstack{\textbf{Total}\\\textbf{Time}}}} \\
\cline{2-6}  
& \textbf{Size} & \textbf{Time(s)}
& \textbf{Size} & \textbf{Time(s)}
& \textbf{Time(s)} 
& \\
\hline
$\Globally \apred$ & 2 &  0.01  &  651 + 642 & 0.23 & 0.35  &  1.34  \\
\hline
$\Finally \appink$ & 5 &  0.01 &  757 + 680  & 0.52  &  0.46 &  1.75 \\ 
\hline
$\mathit{\apcyan \Until \apblue}$ & 7  & 0.04  & 830 + 691 &  0.69 & 0.55  & 2.04 \\ 
\hline
$\mathit{\apblue \Release \apcyan}$ & 7&  0.04   & 814 + 685 & 0.69 & 0.53  & 2.03 \\ 
\hline
$\Finally\Globally \apred$ & 6 & 0.02 &  1,933+1,924 &  0.60& 2.05   & 3.44 \\ 
\hline
$\Globally\Finally \apgreen$ &   7&0.02 &  4,640+2,631  &1.93 & 2.48& 5.20\\ 
\hline
% $ \Globally \apred  \land \Finally \apgreen$ & 9 & 0.34 & 4,640+2,631& 3.78 & 2.71&  7.59 \\ 
% \hline
$\Finally(\apgreen \land \Finally\appink)$ &  33 & 0.40 & 4,856+2,687 &  8.98 &  3.05 & 13.19 \\ 
\hline 
 $\Globally \apred \land (\Finally \appink \land \Finally \apcyan)$  
 &  46  &  51.39 & 7,723+4,144   &  29.66 &  7.54 & 89.36  \\ 
 \hline  
$\Globally \apred \land \Finally(\apgreen\land \Finally\appink)$
& 49  & 53.86  & 7,279+4,030  & 25.49  &  7.06  &  87.18 \\ 
\hline
\end{tabular}

\caption{Sizes (number of states) and construction times (averaged over 10 runs) for the deadlock-free reachable parts of
% the nondeterministic \ABZESN{} automaton
$\BBB_\phi$ and the corresponding parity games. For the games, sizes are reported as the total number of states controlled by Player (angelic nondeterminism in $\BBB_\phi$) and Opponent (demonic nondeterminism in $\SSS$). The table also shows average game construction times for the deadlock-free reachable parts, game-solving times, and total times.
%per specification.
Total times include the construction of the %deadlock-free reachable part of the 
symbolic model $\SSS$, which has 1,089 states and takes an average of 0.77 seconds to build. \label{fig:results}}
% \caption{Sizes (number of states) and construction times (averaged over 10 runs) for the deadlock-free reachable parts of the nondeterministic $\ABZESN{}$ automata $\BBB_\phi$ and the corresponding parity games. 
% %Each cell for $\mathcal{B}$ reports the sizes of both the generalized and nondeterministic automata. 
% For the parity games, sizes are reported as the total number of states controlled by the player (for angelic nondeterminism in $\BBB_\phi$) and the opponent (for demonic nondeterminism in $\SSS$). The table also includes average construction times for the deadlock-free reachable parts of the games, game-solving times, and total times per specification. Total time includes the construction of the deadlock-free reachable part of the symbolic model $\SSS$, which consists of 1,089 states and takes an average of 0.77 seconds to construct. 
% Note that for some specifications, it is possible to verify each subformula in the conjunction separately. However, we verify the entire formula as a whole to demonstrate the feasibility of our approach.
% \label{fig:results}}
\end{figure}

% \subsection{Monitoring}

% If we have space.

% \subsection{Verification}

% - A simple verification application.

% - show the automaton (preferably about 10- 20 states)

\section{Conclusions}
\label{section:conclusions}

We have introduced a novel translation of LTL formulas to \ABZESN{} automata, which is specifically designed for verification of continuous-time systems by abstracting truth values on an interval to four possible patterns.
We have presented a verification algorithm that uses this translation for the abstraction-based verification of nonlinear, nondeterministic, continuous-time, continuous-state systems without global stability assumptions.

% controller synthesis
In the future, we plan to adapt this framework to tackle the symbolic controller synthesis problem.
% , by solving a game built from the specification and the model.
% bottleneck <- just before conclusion
% B, S
 We also want to weaken the constraints imposed on the system, specifically that when an atomic proposition holds, it should continue to hold for a certain amount of time.
 This can be done by introducing other observation patterns, which makes the construction more complex.

% \blind{
\subsubsection{Acknowledgments}
The authors would like to thank Jérémy Dubut for letting them reuse part of his implementation.
% }

\switchversion
{
\bibliographystyle{IEEEtran}
\bibliography{refs}
}
{ 

}

% The preferred spelling of the word ``acknowledgment'' in America is without 
% an ``e'' after the ``g''. Avoid the stilted expression ``one of us (R. B. 
% G.) thanks $\ldots$''. Instead, try ``R. B. G. thanks$\ldots$''. Put sponsor 
% acknowledgments in the unnumbered footnote on the first page.

\switchversion{}{
\appendix

\section{Omitted Proofs}

\subsection{Proof of Lemma~\ref{lem:bounded speed}}
\label{proof:lem:bounded speed}

Given a speed-bounded, $\AP$-separated $\Sigma$, we want to show that Equation~\eqref{eq: Z E} holds for $\tau \leq \inf_{p \in \AP} \inf_{p' \in \AP, p' \neq p} \inf \Delta^{-1}(d_{p,p'})$.
It is enough to show that if $\tau \leq \inf \Delta^{-1}(d_{p,p'})$, then Equation~\eqref{eq: Z E} holds for $p$ and $p'$.

We denote by $\partial p$ the topological boundary of $\setcomp{x \in X}{P(x)(p) = \top}$:
\begin{align*}
    \partial p = \left\{x \in X\,\middle|\, \forall \epsilon > 0. \right. & \left. (\exists y^+ \in X.\, \infinitynorm{x-y^+} < \epsilon \land P(y^+)(p) = \top) \land {} \right. \\
    & \left. (\exists y^- \in X.\, \infinitynorm{x-y^-} < \epsilon \land P(y^-)(p) = \bot)\right\}
\end{align*}
and by $d_p \colon X \to \Rnonneg$ the distance to $\partial p$:
$
    d_p(x) = \inf_{y \in \partial p} \infinitynorm{x-y}.
$
First, by speed-boundedness, it is easy to show that $d_p$ is continuous.
It is similarly easy to show that all trajectories $\sigma$ are continuous, again by speed-boundedness.

Given $t$ such that $P(\sigma(t))(p) \neq P(\sigma(t+\tau))(p)$, we have $d_p(\sigma(t)) d_p(\sigma(t+\tau)) \leq 0$, so by continuity there exists $t_p \in [t, t+\tau]$ such that $d_p(t_p) = 0$, i.e., $\sigma(t_p) \in \partial p$ by definition of $\partial p$.
Similarly, if $P(\sigma(t))(p') \neq P(\sigma(t+\tau))(p')$, there exists $t_{p'} \in [t, t+\tau]$ such that $\sigma(t_{p'}) \in \partial p'$.
Without loss of generality, we assume $t_p \leq t_{p'}$.
Then we have $\sigma(t_{p'}) \in \xi(\sigma(t_p), t_{p'} - t_p)$, so by speed-boundedness
\begin{align*}
    \infinitynorm{\sigma(t_{p'}) - \sigma(t_p)} &< \Delta(t_{p'} - t_p) \leq \Delta(\tau) \leq \Delta(\inf \Delta^{-1}(d_{p,p'})) = d_{p,p'}.
\end{align*}
But this contradicts $\AP$-separatedness if $p \neq p'$, so $p = p'$ as desired.

\subsection{Proof of Lemma~\ref{lem:consistency}}
\label{proof:lem:consistency}

For soundness, we proceed by case distinction on $\odot$ in $\psi = \psi_1 \odot \psi_2$, then on case distinction on the values $\actualchopping{\varsigma}{\tau}(k)(\psi_1)$ and $\actualchopping{\varsigma}{\tau}(k)(\psi_2)$.
\begin{itemize}
    \item If $\psi = \psi_1 \land \psi_2$, we only show that if $\actualchopping{\varsigma}{\tau}(k)(\psi_1) = \actualchopping{\varsigma}{\tau}(k)(\psi_2) = A$, then $\actualchopping{\varsigma}{\tau}(k)(\psi) = A$, the other cases are similar.
    We have that for all $t \in [k \tau, (k+1) \tau]$, $\varsigma, t \vDash \psi_1$ and $\varsigma, t \vDash \psi_2$, so for all $t \in [k \tau, (k+1) \tau]$, $\varsigma, t \vDash \psi$.
    \item If $\psi = \psi_1 \Until \psi_2$, we only show that if $\actualchopping{\varsigma}{\tau}(k)(\psi_1) = A$ and $\actualchopping{\varsigma}{\tau}(k)(\psi_2) = N$, then $\actualchopping{\varsigma}{\tau}(k)(\psi) \in \set{A, N}$, the other cases are simpler.
    We have that for all $t \in [k \tau, (k+1) \tau]$, $\varsigma, t \vDash \psi_1$ and $\varsigma, t \nvDash \psi_2$.
    Therefore, either there exists $t > (k+1) \tau$ such that $\varsigma, t \vDash \psi_2$ and for all $t' \in ((k+1) \tau, t')$, $\varsigma, t' \vDash \psi_1$, in which case $\actualchopping{\varsigma}{\tau}(k)(\psi) = A$,
    or for all $t' > (k+1) \tau$ there exists $t' \in ((k+1) \tau, t')$ such that $\varsigma, t' \nvDash \psi_1$, so $\actualchopping{\varsigma}{\tau}(k)(\psi) = N$.
    \item The other cases are similar.
\end{itemize}

For completeness, it is just a matter of exhibiting signals that have the desired property, which is simple and we do not make explicit.

\subsection{Proof of Lemma~\ref{lem:ap-word-extension-accepting-run}}
\label{proof:lem:ap-word-extension-accepting-run}

First, we exhibit an accepting run. We build $\nu_k(\psi)$ by induction on subformulas $\psi$ and simultaneously prove the following properties:
\begin{align}
    & \text{$\psi = p \in \AP$, then for all $k \in \N$, $\nu_k(\psi) = w_k(p)$}
    \label{eq:ap-word-extension-accepting-run:consistency-aps}
    \\
    & \text{$\psi = \psi_1 \odot \psi_2$, then for all $k \in \N$, $\nu_k(\psi) \in c_\odot(\nu_k(\psi_1), \nu_k(\psi_2))$}
    \label{eq:ap-word-extension-accepting-run:consistency-induction}
    \\
    & \text{for all $k \in \N$, if $\nu_k(\psi) \in \set{E, Z}$,}
    \nonumber \\
    & \qquad \text{then there exists $ p \in \AP \cap \sub(\psi)$ such that $\nu_k(p) = \nu_k(\psi)$}
    \label{eq:ap-word-extension-accepting-run:ZE-subformula} \\
    & \text{for all $k \in \N$, $\nu_k(\psi) \in \set{A, E}$ iff $\nu_{k+1}(\psi) \in \set{A, Z}$}
    \label{eq:ap-word-extension-accepting-run:transition}
\end{align}
\begin{itemize}
    \item
        If $\psi = p$, we define $\nu_k(p) = w_k(p)$.
        Equations~\eqref{eq:ap-word-extension-accepting-run:consistency-aps} and~\eqref{eq:ap-word-extension-accepting-run:ZE-subformula} obviously hold, Equation~\eqref{eq:ap-word-extension-accepting-run:consistency-induction} is void, and Equation~\eqref{eq:ap-word-extension-accepting-run:transition} holds by Equation~\eqref{eq:ap-word-extension-accepting-run:ap-transition}.
    \item
        If $\psi = \psi_1 \land \psi_2$, we define $\nu_k(\psi_1 \land \psi_2) = c_\land(\nu_k(\psi_1), \nu_k(\psi_2))$, by which we mean the unique element of $c_\land(\nu_k(\psi_1), \nu_k(\psi_2))$.
        Equations~\eqref{eq:ap-word-extension-accepting-run:consistency-aps} is void and Equation~\eqref{eq:ap-word-extension-accepting-run:consistency-induction} holds by construction.
        
        We now want to show Equation~\eqref{eq:ap-word-extension-accepting-run:ZE-subformula}, so we assume that $\nu_k(\psi) \in \set{Z, E}$.
        By scrutinizing Fig.~\ref{fig:consistency}, we see that if $\nu_k(\psi) \in \set{Z, E}$, then $\nu_k(\psi_1) = \nu_k(\psi)$ or $\nu_k(\psi_2) = \nu_k(\psi)$ by Lemma~\ref{lem:consistency}.
        Without loss of generality, we assume that $\nu_k(\psi_1) = \nu_k(\psi)$.
        By induction hypothesis, we know that there exists $p \in \sub(\psi_1) \subseteq \sub(\psi)$ such that $\nu_k(p) = \nu_k(\psi_1) = \nu_k(\psi)$ as desired.

        We now want to show Equation~\eqref{eq:ap-word-extension-accepting-run:transition}.
        By scrutinizing Fig.~\ref{lem:consistency}, we see that $\nu_k(\psi)$ is in $\set{A, E}$ iff both $\nu_k(\psi_1)$ and $\nu_k(\psi_2)$ are in $\set{A, E}$.
        By induction hypothesis, this is equivalent to both $\nu_{k+1}(\psi_1)$ and $\nu_{k+1}(\psi_2)$ being in $\set{A, Z}$, which is equivalent to $\nu_{k+1}(\psi)$ being in $\set{A, Z}$ again by scrutinizing Fig.~\ref{lem:consistency}.
    \item
        If $\psi = \psi_1 \lor \psi_2$, we define $\nu_k(\psi_1 \lor \psi_2) = \neg (c_\land(\neg \nu_k(\psi_1), \neg \nu_k(\psi_2)))$, where $\neg \colon \OOO \to \OOO$ is the involution such that $\neg A = N$ and $\neg Z = E$, from which all properties follow directly by the same arguments as above.
    \item
        If $\psi = \psi_1 \Until \psi_2$, we define
        \[
            \nu_k(\psi_1 \Until \psi_2) =
            \begin{cases}
                A & \text{if $\nu_k(\psi_2) = A$ or} \\
                & \text{\phantom{if} $\exists k' > k.\ \nu_{k'}(\psi_2) \neq N$ and $\forall k \leq k'' < k'.\ \nu_{k''}(\psi_1) = A$} \\
                Z & \text{if $\nu_k(\psi_2) = Z$ and} \\
                & \text{\phantom{if} $\forall k' > k.\ (\nu_{k'}(\psi_2) \neq N \Rightarrow \exists k \leq k'' < k'.\ \nu_{k''}(\psi_1) \neq A)$} \\
                E & \text{if $(\nu_k(\psi_2) = E \text{ and } \nu_k(\psi_1) \in \set{E, N})$ or}\\
                & \text{\phantom{if}($\nu_k(\psi_2) = N$, $\nu_k(\psi_1) = E$, and $\exists k' > k.\ \nu_{k'}(\psi_2) \neq N$} \\
                & \text{\phantom{if (}and $\forall k < k'' < k'.\ \nu_{k''}(\psi_1) = A$)} \\
                N & \text{if $\forall k' \geq k.\ \nu_{k'}(\psi_2) \neq N \Rightarrow \exists k \leq k'' < k'.\ \nu_{k''}(\psi_1) \neq A$.}
            \end{cases}
        \]
        Equation~\eqref{eq:ap-word-extension-accepting-run:consistency-aps} is void.
        It is not directly obvious that $\nu_k(\psi)$ is well-defined, so we first show that it is, as well as Equation~\eqref{eq:ap-word-extension-accepting-run:consistency-induction}, by case distinction on $\nu_k(\psi_1)$ and $\nu_k(\psi_2)$.
        \begin{itemize}
            \item
                If $\nu_k(\psi_2) = A$, then the case for $\nu_k(\psi) = A$ holds, the cases for $\nu_k(\psi) \in \set{Z, E}$ obviously do not hold, and the case for $\nu_k(\psi) = N$ does not hold since for $k' = k$, $\nu_{k'}(\psi_2) = A \neq N$, but there is not $k < k'' \leq k' = k$.
                Therefore, $\nu_k(\psi_2) = A$ is well-defined, and Equation~\eqref{eq:ap-word-extension-accepting-run:consistency-induction} holds by scrutinizing Fig.~\ref{fig:consistency}.
            \item
                If $\nu_k(\psi_1) = A$, $\nu_k(\psi_2) = Z$, then the case for $\nu_k(\psi) = A$ holds iff $\exists k' > k.\ \nu_{k'}(\psi_2) \neq N$ and $\forall k \leq k'' < k'.\ \nu_{k''}(\psi_1) = A$, the case for $\nu_k(\psi) = Z$ holds iff $\forall k' > k.\ (\nu_{k'}(\psi_2) \neq N \Rightarrow \exists k \leq k'' < k'.\ \nu_{k''}(\psi_1) \neq A)$, which is the negation of the above, so exactly one of them holds.
                Moreover, the case for $\nu_k(\psi) = E$ obviously does not hold, and the case for $\nu_k(\psi) = N$ does not hold (take $k' = k$).
                Therefore, $\nu_k(\psi) \in \set{A, Z}$ is well-defined and Equation~\eqref{eq:ap-word-extension-accepting-run:consistency-induction} holds.
            \item
                If $\nu_k(\psi_1) = E$, $\nu_k(\psi_2) = Z$, then by induction hypothesis, Equation~\eqref{eq:ap-word-extension-accepting-run:ZE-subformula} holds for $\psi_1$ and $\psi_2$, so there is $p \in \AP \cap \sub(\psi_1)$ such that $w_k(p) = \nu_k(p) = \nu_k(\psi_1) = E$ and $p' \in \AP \cap \sub(\psi_2)$ such that $w_k(p') = \nu_k(p') = \nu_k(\psi_2) = Z$, which contradicts Equation~\eqref{eq:ap-word-extension-accepting-run:no-ZE}, so this case never happens.
            \item
                The other cases are similar to the three cases above.
        \end{itemize}
        To prove Equation~\eqref{eq:ap-word-extension-accepting-run:ZE-subformula}, we know by Lemma~\ref{lem:consistency} that $\nu_k(\psi) \in \set{Z, E}$ implies that $\nu_k(\psi) = \nu_k(\psi_1)$ or $\nu_k(\psi) = \nu_k(\psi_2)$.
        Without loss of generality, if we assume $\nu_k(\psi) = \nu_k(\psi_1)$, then by induction hypothesis there exists $p \in \AP \cap \sub(\psi_1) \subseteq \AP \cap \sub(\psi)$ such that $\nu_k(p) = \nu_k(\psi_1) = \nu_k(\psi)$.
        
        Now, we want to show that Equation~\eqref{eq:ap-word-extension-accepting-run:transition} holds.
        We can show that
        \begin{align*}
            \nu_k(\psi) \in \set{A, E}
            & \iff \nu_k(\psi_2) \in \set{A, E} \lor {} \\
            & \phantom{\iff\,}(\nu_k(\psi_1) \in \set{A, E} \land {}\\
            & \phantom{\iff\, \nu_k} \exists k' > k. (\nu_{k'}(\psi_2) \neq N \land \forall k \leq k'' < k'.\ \nu_{k''}(\psi_1) = A)), \\
            \nu_k(\psi) \in \set{A, Z}
            & \iff \nu_k(\psi_2) \in \set{A, Z} \lor {}\\
            & \phantom{\iff\,} \exists k' > k. (\nu_{k'}(\psi_2) \neq N \land \forall k \leq k'' < k.\ \nu_{k''}(\psi_1) = A).
        \end{align*}
        Therefore,
        \begin{align*}
            \nu_{k+1}(\psi) \in \set{A, Z}
            & \iff \nu_{k+1}(\psi_2) \in \set{A, Z} \lor {} \\
            & \phantom{\iff\,} \exists k' > k+1. (\nu_{k'}(\psi_2) \neq N \land {} \\
            & \phantom{\iff \exists k' > k+1. (}\forall k+1 \leq k'' < k.\ \nu_{k''}(\psi_1) = A) \\
            & \iff \nu_{k}(\psi_2) \in \set{A, E} \lor {} \\
            & \phantom{\iff\,} \exists k' > k+1. (\nu_{k'}(\psi_2) \neq N \land {} \\
            &\phantom{\iff \exists k' > k+1. (} \forall k+1 \leq k'' < k.\ \nu_{k''}(\psi_1) = A)
        \end{align*}
        It is thus obvious that $\nu_k(\psi) \in \set{A, E}$ implies that $\nu_{k+1}(\psi) \in \set{A, Z}$.
        Moreover, if there exists $k' > k+1$ such that $\nu_{k'}(\psi_2) \neq N$ and for all $k+1 \leq k'' < k$, $\nu_{k''}(\psi_1) = A$, then $\nu_k(\psi_1) \in \set{A, E}$ by~\eqref{eq:ap-word-extension-accepting-run:transition} on $\psi_1$, so $\nu_{k+1}(\psi) \in \set{A, Z}$ implies $\nu_k(\psi) \in \set{A, E}$.
    \item
        If $\psi = \psi_1 \Release \psi_2$, we define $\nu_k(\psi_1 \lor \psi_2) = \neg (c_{\Until}(\neg \nu_k(\psi_1), \neg \nu_k(\psi_2)))$, from which all properties follow directly by the same arguments as above.
\end{itemize}
By Equations~\eqref{eq:ap-word-extension-accepting-run:consistency-induction} and~\eqref{eq:ap-word-extension-accepting-run:transition}, $\nu_0 \nu_1 \ldots$ is a run of $\AAA_\phi$, and $\nu_k(p) = w_k(p)$ for all $k \in \N$ and $p \in \AP$ by Equation~\eqref{eq:ap-word-extension-accepting-run:consistency-aps}.

We now prove that $\nu_0 \nu_1 \ldots$ is accepting, i.e., that for each $\psi = \psi_1 \Until \psi_2 \in \sub(\phi)$, it visits $F_\psi$ infinitely often (and similarly for $\psi_1 \Release \psi_2$).
We only prove the case $\psi_1 \Until \psi_2$, the other case is symmetric.
By contradiction, assume that $\nu_0 \nu_1 \ldots$ stops visiting $F_\psi$ after index $k_0$, then for all $k > k_0$, $\nu_k(\psi) = A$ and $\nu_k(\psi_2) = N$.
But $\nu_k(\psi) = A$ and $\nu_k(\psi_2) = N$ iff $\nu_k(\psi_2) = N$ and $\exists k' > k. (\nu_{k'}(\psi_2) \neq N \land \forall k \leq k'' <k'.\ \nu_k(\psi_1) = A)$, so in particular $\nu_k(\psi_2) = N$ for all $k > k_0$ and $\exists k' > k_0 + 1. \nu_{k'}(\psi_2) \neq N$, hence a contradiction, as desired.

Finally, we prove that, if $\nu'_0 \nu'_1 \ldots$ is an accepting run such that for all $k \in \N$ and $p \in \AP$, $\nu'_k(p) = w_k(p)$, then $\nu'_k = \nu_k$.
We prove by induction on $\psi \in \sub(\phi)$ that for all $k \in \N$, $\nu'_k(\psi) = \nu_k(\psi)$.
\begin{itemize}
    \item If $\psi = p$, then the result is obvious.
    \item If $\psi = \psi_1 \land \psi_2$ or $\psi = \psi_1 \lor \psi_2$, then the result holds by induction hypothesis on $\psi_1$ and $\psi_2$, using the fact that $c_\land(o_1, o_2)$ and $c_\lor(o_1, o_2)$ contain a unique element.
    \item If $\psi = \psi_1 \Until \psi_2$, then by induction hypothesis $\nu'_k(\psi_1) = \nu_k(\psi_1)$ and $\nu'_k(\psi_2) = \nu_k(\psi_2)$, then we proceed by case distinction on $\nu_k(\psi_1)$ and $\nu_k(\psi_2)$.
    In most cases, $c_{\Until}(\nu_k(\psi_1), \nu_k(\psi_2))$ contains only one element, so the result holds directly, so we only detail the other cases.
    \begin{itemize}
        \item If $\nu_k(\psi_1) = A$, $\nu_k(\psi_2) = Z$, and $\nu_k(\psi) = A$, then because $\nu_k(\psi) = A$, we know that either $\nu_k(\psi_2) = A$ (which is not true) or there exists $k' > k$ such that $\nu_{k'}(\psi_2) \neq N$ and for all $k \leq k'' < k'$, $\nu_{k''}(\psi_1) = A$.
        Let $k'$ be the smallest such index, then for all $k < k'' < k'$, $\nu_{k''}(\psi_2) = N$, and $\nu_{k'}(\psi_2) \in \set{E, N}$ by~\eqref{eq:ap-word-extension-accepting-run:transition}, so $\nu_{k'}(\psi_2) = E$.
        Therefore, we have for all $k < k'' < k'$, $\nu'_{k''}(\psi_1) = A$ and $\nu'_{k''}(\psi_2) = N$, so $\nu'_{k''}(\psi) \in \set{A, N}$.
        We also have $\nu'_{k'}(\psi_1) \in \set{A, Z}$ by~\eqref{eq:ap-word-extension-accepting-run:transition} and $\nu'_{k'}(\psi_2) = E$, so $\nu'_{k'}(\psi_1) = A$ by~\eqref{eq:ap-word-extension-accepting-run:no-ZE}.
        Therefore, $\nu'_{k'}(\psi) = A$, whence $\nu'_{k''}(\psi) = A$ for all $k < k'' < k'$ (by induction on $k' - k''$).
        Finally, because $\nu'_k(\psi) \in \set{A, Z}$, $\nu'_k(\psi) = A$ by~\eqref{eq:ap-word-extension-accepting-run:transition}.
        \item If $\nu_k(\psi_1) = A$, $\nu_k(\psi_2) = Z$, and $\nu_k(\psi) = Z$, then because $\nu_k(\psi) = Z$, we know that $\nu_k(\psi_2) = Z$ and for all $k' > k$, if $\nu_{k'}(\psi_2) \neq N$, then there exists $k \leq k'' < k'$ such that $\nu_{k''}(\psi_1) \neq A$.
        \begin{itemize}
            \item If for all $k' > k$, $\nu_{k'}(\psi_2) = N$, and we assume that $\nu'_k(\psi) \neq Z$, then $\nu'_k(\psi) = A$, then because $\nu'_{k'}(\psi_1) = N$, we have $\nu'_{k'}(\psi) \in \set{A, E, N}$ according to Lemma~\ref{lem:consistency}.
            By induction and~\eqref{eq:ap-word-extension-accepting-run:transition}, we have that $\nu'_{k'}(\psi) = A$.
            Therefore, we have for all $k' > k$ that $\nu'_{k'}(\psi) = A$ and $\nu'_{k'}(\psi_2) = N$, which contradicts the fact that $\nu'_0 \nu'_1 \ldots$ is accepting.
            \item Otherwise, let $k' > k$ be the minimal index such that $\nu_{k'}(\psi_2) \neq N$, then we have that for all $k < k'' < k'$, $\nu_{k''}(\psi_2) = N$, and $\nu_{k'}(\psi_2) \in \set{E, N}$ by~\eqref{eq:ap-word-extension-accepting-run:transition}, hence $\nu_{k'}(\psi_2) = E$.
            Let $k \leq k'' < k'$ be the minimal index such that $\nu_{k''}(\psi_1) \neq A$, then for all $k \leq k''' < k''$, $\nu_{k'''}(\psi_1) = A$, and $\nu_{k''}(\psi_1) \in \set{A, Z}$ by~\eqref{eq:ap-word-extension-accepting-run:transition}, hence $\nu_{k''}(\psi_1) = Z$.
            Now, for all $k < k''' < k''$, $\nu'_{k'''}(\psi_1) = A$ and $\nu'_{k'''}(\psi_2) = N$, hence $\nu'_{k'''}(\psi) \in \set{A, N}$ by Lemma~\ref{lem:consistency}.
            If we assume that $\nu'_k(\psi) = A$, then $\nu'_{k'''}(\psi) = A$ by induction using~\eqref{eq:ap-word-extension-accepting-run:transition}.
            Moreover, $\nu'_{k''}(\psi_1) = Z$ and $\nu'_{k''}(\psi_2) = N$, hence $\nu'_{k''}(\psi) = N$, which contradicts $\nu'_{k''-1} = A$.
        \end{itemize}
        \item If $\nu_k(\psi_1) = A$, $\nu_k(\psi_2) = N$, and $\nu_k(\psi) = A$, then because $\nu_k(\psi) = A$, we know that either $\nu_k(\psi_2) = A$ (which is not true) or there exists $k' > k$ such that $\nu_{k'}(\psi_2) \neq N$ and for all $k \leq k'' < k'$, $\nu_{k''}(\psi_1) = A$.
        Let $k'$ be the minimal such index, then for all $k < k'' < k'$, $\nu_{k''}(\psi_2) = N$, and $\nu_{k'}(\psi_2) \in \set{E, N}$ by~\eqref{eq:ap-word-extension-accepting-run:transition}, hence $\nu_{k'}(\psi_2) = E$.
        Moreover, $\nu_{k'}(\psi_1) \in \set{A, Z}$ by Lemma~\ref{lem:consistency}, hence $\nu_{k'}(\psi_1) = A$ by~\eqref{eq:ap-word-extension-accepting-run:no-ZE}.
        We thus have that, $\nu'_{k'}(\psi_1) = A$, $\nu'_{k'}(\psi_2) = E$, so $\nu'_{k'}(\psi) = A$ by Lemma~\ref{lem:consistency}, and for all $k < k'' < k'$, $\nu'_{k''}(\psi_1) = A$, $\nu'_{k''}(\psi_2) = N$, so $\nu'_{k''}(\psi) \in \set{A, N}$, thus $\nu'_{k''}(\psi) = A$ by induction on $k' - k''$ using~\eqref{eq:ap-word-extension-accepting-run:transition}.
        By Lemma~\ref{lem:consistency}, $\nu'_k(\psi) \in \set{A, N}$, hence $\nu'_k(\psi) = A$ by~\eqref{eq:ap-word-extension-accepting-run:transition}.
        \item If $\nu_k(\psi_1) = A$, $\nu_k(\psi_2) = N$, and $\nu_k(\psi) = N$, then because $\nu_k(\psi) = N$, we know that for all $k' > k$, if $\nu_{k'}(\psi_2) \neq N$, then there exists $k \leq k'' < k'$ such that $\nu_{k''}(\psi_1) \neq A$.
        \begin{itemize}
            \item If for all $k' > k$, $\nu_{k'}(\psi_2) = N$, then for all $k' > k$, $\nu'_{k'}(\psi) \neq Z$ by Lemma~\ref{lem:consistency}.
            If $\nu_k(\psi) = A$, then for all $k' > k$, $\nu'_{k'}(\psi) = A$ by induction on $k'$ using~\eqref{eq:ap-word-extension-accepting-run:transition}, therefore $\nu'_0 \nu'_1 \ldots$ is not accepting.
            \item Otherwise, let $k' > k$ be the minimal index such that $\nu_{k'}(\psi_2) \neq N$, then for all $k \leq k'' < k'$, $\nu_{k''}(\psi_2) = N$, and $\nu_{k'}(\psi_2) \in \set{E, N}$ by~\eqref{eq:ap-word-extension-accepting-run:transition}, hence $\nu_{k'}(\psi_2) = E$.
            Moreover, let $k \leq k'' < k'$ by the minimal index such that $\nu_{k''}(\psi_1) \neq A$, then for all $k \leq k''' < k''$, $\nu_{k'''}(\psi_1) = A$, and $\nu_{k''}(\psi_1) \in \set{A, Z}$ by~\eqref{eq:ap-word-extension-accepting-run:transition}, hence $\nu_{k''}(\psi_1) = Z$.
            Therefore, $\nu'_{k''}(\psi_1) = Z$ and $\nu'_{k''}(\psi_2) = N$, so $\nu'_{k''}(\psi) = N$ by Lemma~\ref{lem:consistency}.
            For all $k \leq k''' < k''$, $\nu'_{k'''}(\psi_1) = A$ and $\nu'_{k'''}(\psi_2) = N$, so $\nu'_{k'''}(\psi) \in \set{A, N}$, and therefore $\nu'_{k'''}(\psi) = N$ by induction on $k'' - k'''$, hence $\nu'_k(\psi) = N$.
            \qedhere
        \end{itemize}
    \end{itemize}
\end{itemize}

\subsection{Proof of Corollary~\ref{cor:accepting-word-signal}}
\label{proof:cor:accepting-word-signal}

It is enough to show that $\actualchopping{\varsigma}{\tau}$ satisfies the definitions in the proof of Lemma~\ref{lem:ap-word-extension-accepting-run}, which we show for all $k \in \N$ and $\psi \in \sub(\phi)$ by induction on $\psi$.
\begin{itemize}
    \item
        If $\psi = p$, then the result holds directly.
    \item
        If $\psi = \psi_1 \land \psi_2$, then
        \begin{align*}
            \actualchopping{\varsigma}{\tau}(k)(\psi) = A
            & \iff \forall t \in [k \tau, (k+1) \tau].\, \varsigma, t \vDash \psi \\
            & \iff \forall t \in [k \tau, (k+1) \tau].\, \varsigma, t \vDash \psi_1 \land \psi_2 \\
            & \iff \forall t \in [k \tau, (k+1) \tau].\, \varsigma, t \vDash \psi_1 \text{ and } t \vDash \psi_2 \\
            & \iff \forall t \in [k \tau, (k+1) \tau].\, \varsigma, t \vDash \psi_1 \text{ and } \\
            & \phantom{\iff\ \ } \forall t \in [k \tau, (k+1) \tau].\, \varsigma, t \vDash \psi_2 \\
            & \iff \actualchopping{\varsigma}{\tau}(k)(\psi_1) = A \text{ and } \actualchopping{\varsigma}{\tau}(k)(\psi_2) = A
        \end{align*}
        as desired.
        The other cases are similar.
    \item
        If $\psi = \psi_1 \Until \psi_2$, then
        \begin{align}
            \actualchopping{\varsigma}{\tau}(k)(\psi) = A
            & \iff \forall t \in [k \tau, (k+1) \tau].\, \varsigma, t \vDash \psi
            \nonumber \\
            & \iff \forall t \in [k \tau, (k+1) \tau].\, \varsigma, t \vDash \psi_1 \Until \psi_2
            \nonumber \\
            & \iff \forall t \in [k \tau, (k+1) \tau].\,
            \nonumber \\
            & \phantom{\iff\,\quad} \exists t' \geq t.\, (\varsigma, t' \vDash \psi_2 \land \forall t'' \in [t, t').\, \varsigma, t'' \vDash \psi_1).
            \label{eq:accepting-word-signal:until:A:lhs}
        \end{align}
        On the other hand
        \begin{align}
            & \actualchopping{\varsigma}{\tau}(k)(\psi_2) = A \text{ or } \exists k' > k. (\actualchopping{\varsigma}{\tau}(k')(\psi_2) \neq N \text{ and } \forall k \leq k'' < k'.\, \actualchopping{\varsigma}{\tau}(k'') = A) \nonumber \\
            &\iff \forall t \in [k \tau, (k+1) \tau].\, \varsigma, t \vDash \psi_2 \text{ or }
            \nonumber \\
            &\phantom{\iff\ \ } \exists k' > k. (\exists t' \in [k' \tau, (k' + 1) \tau].\, \varsigma, t' \vDash \psi_2 \text{ and }
            \nonumber \\
            &\phantom{\iff\ \ \exists k' > k. (} \forall k \leq k'' < k'.\, \forall  t'' \in [k'' \tau, (k''+1) \tau].\, \varsigma, t'' \vDash \psi_1).
            \label{eq:accepting-word-signal:until:A:rhs}
        \end{align}
        We want to show that~\eqref{eq:accepting-word-signal:until:A:lhs} and~\eqref{eq:accepting-word-signal:until:A:rhs} are equivalent.
        We distinguish four cases:
        \begin{itemize}
            \item If $\varsigma, k \tau \vDash \psi_2$ and $\varsigma, (k+1) \tau \vDash \psi_2$, then by Assumption~\ref{asm: restriction BSZE}.\eqref{eq: no B,S}, for all $t \in [k \tau, (k+1) \tau]$, $\varsigma, t \vDash \psi_2$, so both~\eqref{eq:accepting-word-signal:until:A:lhs} and~\eqref{eq:accepting-word-signal:until:A:rhs} hold.
            \item If $\varsigma, k \tau \vDash \psi_2$ and $\varsigma, (k+1) \tau \nvDash \psi_2$, then by Assumption~\ref{asm: restriction BSZE}.\eqref{eq: Z E}, there exists $t \in (k \tau, (k+1) \tau)$ such that $\varsigma, t' \vDash \psi_2$ for all $t' \in [k \tau, t)$ and $\varsigma, t' \nvDash \psi_2$ for all $t' \in (t, (k+1) \tau]$.
            If~\eqref{eq:accepting-word-signal:until:A:lhs} holds, then there exists $t' > (k+1) \tau$ such that $\varsigma, t' \vDash \psi_2$ and $\varsigma, t'' \vDash \psi_1$ for all $t'' \in [t, t')$.
            Therefore so does~\eqref{eq:accepting-word-signal:until:A:rhs} by taking $k' = \lceil t / \tau \rceil$ and $t' = k' \tau$.

            Conversely, if~\eqref{eq:accepting-word-signal:until:A:rhs} holds, then there exists $k' > k$ and $t' \in [k' \tau, (k'+1) \tau]$, $\varsigma, t' \vDash \psi_2$, and for all $k \leq k'' < k'$ and $t'' \in [k'' \tau, (k''+1) \tau]$, $\varsigma, t'' \vDash \psi_1$.
            In particular, $\varsigma, ((k'-1)+1) \tau \vDash \psi_1$, so $\actualchopping{\varsigma}{\tau}(k')(\psi_1) \in \set{A, Z}$.
            If $\actualchopping{\varsigma}{\tau}(k')(\psi_1) = A$, then~\eqref{eq:accepting-word-signal:until:A:lhs} also directly holds.
            If $\actualchopping{\varsigma}{\tau}(k')(\psi_1) = Z$, then by Assumption~\ref{asm: restriction BSZE}.\eqref{eq: Z E}, because $\varsigma, t' \vDash \psi_2$, $\varsigma, t'' \vDash \psi_2$ for all $t'' \in [k' \tau, t']$ (because $\actualchopping{\varsigma}{\tau}(k')(\psi_2) \neq E$ and formulas can only change values at one time during a time interval of length $\tau$).
            Therefore,~\eqref{eq:accepting-word-signal:until:A:lhs} holds.
            \item The other cases are similar to one of the two cases above.
        \end{itemize}
        The other cases are similar.
    \item
        If $\psi = \psi_1 \lor \psi_2$ or $\psi = \psi_1 \Release \psi_2$, the result follows using the same arguments as above.
        \qedhere
\end{itemize}
}

\end{document}